\documentclass[a4paper,USenglish,cleveref, autoref, thm-restate]{lipics-v2021}
\nolinenumbers
\usepackage{graphicx}
\usepackage{multicol}

\newcommand{\para}[1]{\subparagraph*{#1}}
\newtheorem{fact}[theorem]{Fact}
\newtheorem{problem}[theorem]{Problem}
\crefname{lemma}{Lemma}{Lemmata}

\newcommand{\rev}{\overleftarrow}
\newcommand{\inv}{\overline}

\newcommand{\N}{\mathbb{N}}

\newcommand{\R}{\mathbb{R}}

\newcommand{\eps}{\varepsilon}

\newcommand{\I}{\mathcal{I}}
\newcommand{\J}{\mathcal{J}}
\newcommand{\X}{\mathcal{X}}

\newcommand{\floor}[1]{\left\lfloor{#1}\right\rfloor}

\newcommand{\xmid}{x_\mathsf{{mid}}}
\newcommand{\ymid}{y_\mathsf{{mid}}}

\newcommand{\last}{\mathsf{Last}}
\newcommand{\dist}{\mathsf{dist}}
\newcommand{\HDD}{\mathsf{HDD}}
\newcommand{\LCS}{\mathsf{LCS}}
\newcommand{\SETH}{\mathsf{SETH}}
\newcommand{\cost}{\mathsf{cost}}
\newcommand{\remove}{\mathsf{remove}}

\newcommand{\HC}{\mathsf{HCD}}

\newcommand{\IL}{I_\mathsf{L}}
\newcommand{\IR}{I_\mathsf{R}}
\newcommand{\PL}{P_\mathsf{L}}
\newcommand{\PR}{P_\mathsf{R}}
\newcommand{\SL}{\mathsf{Sync_L}}
\newcommand{\SR}{\mathsf{Sync_R}}

\newcommand{\LEFT}{\mathsf{left}}
\newcommand{\RIGHT}{\mathsf{right}}

\newcommand{\Pairs}{\mathsf{Pairs}}

\newcommand{\leftP}{\mathsf{left}^P}
\newcommand{\leftI}{\mathsf{left}^I}
\newcommand{\rightP}{\mathsf{right}^P}
\newcommand{\rightI}{\mathsf{right}^I}

\newcommand{\EL}{E_\mathsf{L}}
\newcommand{\ER}{E_\mathsf{R}}
\newcommand{\info}{\mathsf{info}}
\newcommand{\prot}{\mathsf{prot}}
\newcommand{\p}{\mathsf{p}}
\renewcommand{\i}{\mathsf{i}}
\newcommand{\sync}{\mathsf{sync}}

\renewcommand{\int}{\mathsf{int}}
\newcommand{\ext}{\mathsf{ext}}
\renewcommand{\middle}{\mathsf{mid}}

\newcommand{\per}{\mathsf{per}}

\newcommand{\len}{\mathsf{len}}

\newcommand{\Fib}{\mathsf{Fib}}

\newcommand{\match}{\mathsf{match}}
\newcommand{\mis}{\mathsf{mis}}

\crefname{enumi}{Step}{Steps}

\title{Hairpin Completion Distance Lower Bound}
\author{Itai Boneh}{Reichman University and University of Haifa, Israel}{itai.bone@biu.ac.il}{https://orcid.org/0009-0007-8895-4069}{supported by Israel Science Foundation grant 810/21.}

 \author{Dvir Fried}{Bar Ilan University, Israel}{friedvir1@gmail.com}{https://orcid.org/0000-0003-1859-8082}{supported by ISF grant no. 1926/19, by a BSF grant 2018364, and by an ERC grant MPM under the EU's Horizon 2020 Research and Innovation Programme (grant no. 683064).}

  \author{Shay Golan}{Reichman University and University of Haifa, Israel}{golansh1@biu.ac.il}{https://orcid.org/0000-0001-8357-2802}{supported by Israel Science Foundation grant 810/21.}

    \author{Matan Kraus}{Bar Ilan Univesity, Israel}{matan3@gmail.com}{https://orcid.org/0000-0002-2989-1113}{supported by the ISF grant no. 1926/19, by the BSF grant 2018364, and by the ERC grant MPM under the EU's Horizon 2020 Research and Innovation Programme (grant no. 683064).}

\authorrunning{Boneh et al.}

\Copyright{Jane Open Access and Joan R. Public} 

\ccsdesc[500]{Theory of computation~Design and analysis of algorithms}

\keywords{Fine-grained complexity, Hairpin completion, LCS}

\category{} 
\relatedversion{}
\EventEditors{John Q. Open and Joan R. Access}
\EventNoEds{2}
\EventLongTitle{42nd Conference on Very Important Topics (CVIT 2016)}
\EventShortTitle{CVIT 2016}
\EventAcronym{CVIT}
\EventYear{2016}
\EventDate{December 24--27, 2016}
\EventLocation{Little Whinging, United Kingdom}
\EventLogo{}
\SeriesVolume{42}
\ArticleNo{23}

\hideLIPIcs  

\begin{document}

\maketitle

\begin{abstract}
Hairpin completion, derived from the hairpin formation observed in DNA biochemistry, is an operation applied to strings, particularly useful in DNA computing.
Conceptually, a right hairpin completion operation transforms a string $S$ into $S\cdot S'$ where $S'$ is the reverse complement of a prefix of $S$.
Similarly, a left hairpin completion operation transforms a string $S$ into $S'\cdot S$ where $S'$ is the reverse complement of a suffix of $S$.
The hairpin completion distance from $S$ to $T$ is the minimum number of hairpin completion operations needed to transform $S$ into $T$.
Recently Boneh et al.~\cite{BFMP23} showed an $O(n^2)$ time algorithm for finding the hairpin completion distance between two strings of length at most $n$.
In this paper we show that for any $\varepsilon>0$ there is no $O(n^{2-\varepsilon})$-time algorithm for the hairpin completion distance problem unless the Strong Exponential Time Hypothesis (SETH) is false.
Thus, under SETH, the time complexity of the hairpin completion distance problem is quadratic, up to sub-polynomial factors.
\end{abstract}
\clearpage
\setcounter{page}{1} 

\section{Introduction}\label{sec:intro}
Hairpin completion~\cite{CMM06}, derived from the hairpin formation observed in DNA biochemistry, is an operation applied to strings, particularly useful in DNA computing~\cite{KMT07,KKLST06,KKTS05,KLST06}.
Consider a sequences over an alphabet $\Sigma$ with involution $\mathsf{Inv}:\Sigma\to\Sigma$ assigning for every $\sigma\in\Sigma$ an inverse symbol $\inv\sigma$.
For a string $S\in\Sigma^*$, a left hairpin completion transforms $S$ into $\rev{S'}\cdot S$, where $S'$ is a suffix of $S$, and for any $X\in\Sigma^*$ we define $\rev{X}=\inv{X[|X|]}\cdot\inv{X[|X|-1]}\cdot\ldots\cdot\inv{X[1]}$.
This operation can only be applied under the restriction that the suffix $S'$ is preceded by the symbol $\inv{S[1]}$.
Similarly, a right hairpin completion transforms $S$ into $S\cdot\rev{S'}$ where $S'$ is a prefix of $S$ followed by $\inv{S[|S|]}$.

Several problems regarding hairpin completion were studied~\cite{MANEA10,MMM09,MMM10,MANEA09,BFMP23}.
In this paper, we consider the hairpin completion distance problem.
In this problem, we are given two strings $x$ and $y$ and our goal is to compute the minimum number of hairpin completion operations one has to apply on $y$ to transform $y$ into $x$, or to report that there is no sequence of hairpin completion operation can turn $y$ into $x$.
In 2009, Manea, Martín-Vide and Mitrana~\cite{MMM09} proposed the problem and introduced a cubic time $O(n^3)$ algorithm (where $n=|x|$).
Later Manea~\cite{MANEA10} introduced a faster algorithm that runs in $O(n^2\log n)$ time.
Recently, Boneh et al.~\cite{BFMP23} showed that the time complexity of the problem is $O(n^2)$.
Moreover, Boneh et al. posed the following open problem.
\begin{problem}\label{prob:lb}
    Can one prove a lower bound for hairpin completion distance computation that matches the $O(n^2)$ upper bound?
\end{problem}
In this paper, we show that for every $\eps > 0$, there is no $O(n^{2-\eps})$ time algorithm for computing the hairpin completion distance from $y$ to $x$, unless the Strong Exponential Time Hypothesis ($\SETH$)~\cite{CIP09} is false.
Thus, we provide a conditional lower bound matching the upper bound of \cite{BFMP23} up to sub-polynomial factors.
\begin{theorem}\label{thm:reduction}
    Let $\eps>0$.
    If there is an algorithm that computes the hairpin completion distance from $y$ to $x$ in $O(|x|^{2-\varepsilon})$ time, then $\SETH$ is false.
    This holds even if the input strings are over an alphabet of size $4$.
\end{theorem}

We note that due to the relationship between hairpin operations and DNA biochemistry, a typical output for a hairpin-related problem is over the alphabet $\{ A,C,G,T\}$ of size $4$.
Hence, our lower bound applies to a natural set of practical inputs.

\cref{thm:reduction} is proven by reducing Longest Common Subsequence ($\LCS$) problem to the hairpin completion distance problem.
Namely, for two ternary strings $S$ and $T$, we show a linear time construction of a pair of strings $x$ and $y$ such that $\LCS(S,T)$ can be computed in linear time from the hairpin completion distance from $y$ to $x$.
The hardness of hairpin completion computation follows from the conditional lower bound on the $\LCS$ problem~\cite{BK18,ABVW15}.
We refer the reader to \cref{sec:reduction} where we introduce the reduction and to \cref{sec:overview} where we provide a high-level discussion regarding the correctness of our construction.

\section{Preliminaries}\label{sec:prelim}

For $i,j\in\N$ let $[i..j] =\{k \in \N \mid i\le k \le j \}$.
We denote $[i]=[1..i]$.

A string $S$ over an alphabet $\Sigma$ is a sequence of characters $S=S[1]S[2]\ldots S[|S|]$.
For $i,j\in[|S|]$, we call $S[i..j]=S[i]S[i+1]\ldots S[j]$ a \emph{substring} of $S$. If $i=1$, $S[i..j]$ is a prefix of $S$, and if $j=|S|$, $S[i..j]$ is a suffix of $S$.
Let $x$ and $y$ be two strings over an alphabet $\Sigma$.
$x\cdot y$ is the concatenation of $x$ and $y$.
For strings $x_1,x_2, \ldots x_m$, we denote as $\bigodot_{i=1}^{m} = x_1 \cdot x_2 \cdot \ldots \cdot x_m$.
For a string $x$ and $k\in\N$ we write the concatenation of $x$ to itself $k$ times as $x^k$.
For a symbol $\sigma \in \Sigma$, we denote as $\#_\sigma(x) =| \{i \in [|x|] \mid x[i] = \sigma \}|$ the number of occurrences of $\sigma$ in $x$.
We say that a string $y$ occurs in $x$ (or that $x$ contains an occurrence of $y$) if there is an index $i\in [|x| - |y| + 1]$ such that $x[i..i+|y| -1] = y$.

For two sets of strings $\mathcal{S}$ and $\mathcal{T}$, we define the set of strings $\mathcal{S} * \mathcal{T} = \{ s\cdot x \cdot t \mid s\in \mathcal{S} , x\in \Sigma^*,t\in \mathcal{T}\}$.
We use the notations $\mathcal{S} * = \mathcal{S} * \Sigma^*$ and $* \mathcal{S} = \Sigma^* * \mathcal{S}$.
When using $*$ notation, we sometimes write $s \in \Sigma^*$ to denote the set $\{s\}$ (for example, $0*$ is the set of all strings starting with $0$).

\para{Hairpin Operations.}
Let $\mathsf{Inv}: \Sigma \to \Sigma$ be a permutation on $\Sigma$.
We say that $\mathsf{Inv}$ is an \textit{inverse function} on $\Sigma$ if $\mathsf{Inv} = \mathsf{Inv}^{-1}$ and $\mathsf{Inv}(\sigma) \neq \sigma$ for every $\sigma \in \Sigma$.
Throughout this paper (except for \cref{app:equivalence}), we discuss strings over alphabet $\Sigma = \{ 0,1\}$ with $\mathsf{Inv}(\sigma) = 1-\sigma$.
For every symbol in $\sigma$, we denote $\inv\sigma = \mathsf{Inv}(\sigma)$.
We further extend this notation to strings by denoting $\inv x = \inv{x[1]}\cdot \inv{x[2]} \cdot \ldots \cdot \inv{x[|x|]}$.
We denote $\rev{x} = \inv{x}[|x|] \cdot \inv{x}[|x|-1] \cdot \ldots \cdot \inv{x}[2] \cdot \inv{x}[1]$.

We define several types of hairpin operations that can be applied to a string over $\Sigma$ with an inverse function on $\Sigma$.
In \cite{CMM06}, hairpin operations are defined as follows.
\begin{definition}[Hairpin Operations]\label{def:hairpin}
    Let $S\in\Sigma^*$.
    A right hairpin completion of length $\ell \in [|S|]$ transforms $S$ into $S\cdot\rev{S[1..\ell]}$.
    A right hairpin completion operation of length $\ell$ can be applied on $S$ only if $S[\ell +1 ] = \inv{S[|S|]}$.
    Similarly, a left hairpin completion of length $\ell$ transforms $S$ into $\rev{S[|S|-\ell+1..|S|]}\cdot S$.
    A left hairpin completion of length $\ell$ can be applied to $S$ only if $S[|S|-\ell]  = \inv{S[1]}$.
    A right (resp. left) hairpin deletion operation of length $\ell \in [\floor{\frac{|S|}{2}}]$ transforms a string $S$ into a prefix (resp. suffix) $S'$  of $S$ such that $S$ can be obtained from $S'$ by a valid right (resp. left) hairpin completion of length $\ell$.
\end{definition}

Throughout this paper, we use the following modified definition of hairpin operation, which removes the constraints regarding $S[1]$ and $S[|S|]$.

\begin{definition}[Hairpin Operations, Modified definition]\label{def:mod_hairpin}
    Let $S\in\Sigma^*$.
    A right hairpin completion of length $\ell \in [|S|]$ transforms $S$ into $S\cdot\rev{S[1..\ell]}$.
    Similarly, a left hairpin completion of length $\ell$ transforms $S$ into $\rev{S[|S|-\ell+1..|S|]}\cdot S$.
    A right (resp. left) hairpin deletion of length $\ell \in [\floor{\frac{|S|}{2}}]$ operation transforms a string $S$ into a prefix (resp. suffix) $S'$  of $S$ such that $S$ can be obtained from $S'$ by a valid right (resp. left) hairpin completion of length $\ell$.
\end{definition}
We highlight that the modified definition is \textit{not} equivalent to the definition of \cite{CMM06}.
Even though the paper is phrased in terms of the modified definition, we emphasize that \cref{thm:reduction} is correct with respect to both definitions.
In \cref{sec:reduction}, we discuss the machinery required to make our hardness result applicable to \cref{def:hairpin}. The complete details for bridging this gap are developed in \cref{app:equivalence}.

Let $x$ and $y$ be two strings.
We denote by $\HDD(x,y)$ (resp. $\HC(x,y)$) the minimum number of hairpin deletion (resp. completion) operations required to transform $x$ into $y$, counting both left and right operations.
Note that $\HDD(x,y) = \HC(y,x)$.

For the sake of analysis, we define the following graph.
\begin{definition}[Hairpin Deletion Graph]\label{def:graph}
    For a string $x$ the \emph{Hairpin Deletion Graph} $G_x=(V,E)$ is defined as follows.
   $V$ is the set of all substrings of $x$, and $(u,v)\in E$ if $v$ can be obtained from $u$ in a single hairpin deletion operation. 
\end{definition}
We define the distance between two vertices $s$ and $t$ in a graph $G$ (denoted as $\dist_G(s,t)$) to be the minimal length (number of edges) of a path from $s$ to $t$ in $G$ (of $\infty$ if there is no such path).
Note that for a source string $x$ and a destination string $y$, it holds that $\HDD(x,y) = \dist_{G_x}(x,y)$.
We distinguish between two types of edges outgoing from $x[i..j]$.
An edge of the form $x[i..j]\to x[i+\ell..j]$ for some $\ell\in\N$ is called a \emph{left edge} and it corresponds to a left hairpin deletion operation of length $\ell$.
Similarly, an edge of the form $x[i..j]\to x[i..j-\ell]$ for some $\ell\in\N$ is called a \emph{right edge} and it corresponds to a right hairpin deletion operation of length $\ell$.
When a path $p$ in $G_x$ traverses a left (resp. right) edge outgoing from $v$, we say that $p$ applies a left (resp. right) hairpin deletion to $v$.
For a path $p$ we denote by $\cost(p)$ the length of $p$. 

\para{Hairpin Deletion.}
Since the paper makes intensive use of hairpin deletion notations, we introduce an alternative, more intuitive definition for hairpin deletion, equivalent to \cref{def:mod_hairpin}.
For a string $S$, if for some $\ell \in [\floor{\frac{|S|}{2}}]$ we have $S[1..\ell]=\rev{S[|S|-\ell+1..|S|]}$ then a left (resp. right) hairpin deletion operation transforms $S$ into $S[\ell+1..|S|]$  (resp. $S[1.. |S|-\ell]$).
In particular, if $S[1]\ne\rev{S[|S|]}$ then there is no valid hairpin deletion operation on $S$.

\para{Longest Common Subsequence.} 
A subsequence of a string $S$ of length $n$ is a string $X$ of length $\ell$ such that there is an increasing sequence $1\le i_1<i_2< \ldots i_\ell \le n$ satisfying $X[k] = S[i_k]$ for every $k\in [\ell]$.
For two strings $S$ and $T$, a string $X$ is a common subsequence of $S$ and $T$ if $X$ is a subsequence of both $S$ and $T$.
The $\LCS$ problem is, given two strings $S$ and $T$  of length at most $n$, compute the maximum \emph{length} of a common subsequence of $S$ and $T$, denoted as $\LCS(S,T)$.

Bringmann and K{\"{u}}nnemann~\cite{BK18} have shown the following.
\begin{fact}[Hardness of $\LCS$]\label{fact:LCSisHard}
    For every $\eps > 0$, there is no $O(n^{2-\eps})$-time algorithm that solves the $\LCS$ problem for ternary input strings unless $\SETH$ is false.
\end{fact}

\para{Fibonacci sequence.} The Fibonacci sequence is defined as follows.
$\Fib(0)=1$, $\Fib(1)=1$ and for all integer $i>1$ we have $\Fib(i)=\Fib(i-1)+\Fib(i-2)$.
The inverse function $\Fib^{-1} : \R \to \N $ is defined as 
$\Fib^{-1}(x) = \min\{y\in\N \mid \Fib(y) \ge x\}$.

\section{The Reduction}\label{sec:reduction}
Here we introduce a reduction from the $\LCS$ problem on ternary strings. 
We also provide in \cref{sec:overview} a high-level discussion of why the reduction should work.

We present a linear time algorithm such that given two strings $S, T\in\{0,1,2\}^*$,  constructs two (binary) strings $x$ and $y$ with $|x|=O(|S|+|T|)$ and $|y|=O(1)$.
The strings $x$ and $y$ have the property that $\HDD(x, y)$ can be used to infer $\LCS(S,T)$ in linear time.
Thus, by \cref{fact:LCSisHard}, we deduce that any algorithm computing $\HDD(x,y)$ cannot have running time $O(|x|^{2-\eps})$ for any $\eps>0$ (assuming $\SETH$).

We use several types of gadgets.
Let:
\begin{multicols}{2}

\begin{itemize}
    \item $\IL(0)=(010^3)^{\i_0}$
    \item $\IL(1)=(010^5)^{\i_1}$
    \item $\IL(2)=(010^7)^{\i_2}$
    \item $\PL=(010^9)^{\p}$
    \item $\SL=01$
    \item $\IR(0)=(\inv{0^310})^{\i_0}=\rev{\IL(0)}$
    \item $\IR(1)=(\inv{0^510})^{\i_1}=\rev{\IL(1)}$
    \item $\IR(2)=(\inv{0^710})^{\i_2}=\rev{\IL(2)}$
    \item $\PR=(\inv{0^910})^{\p}=\rev{\PL}$
    \item $\SR=\inv{010}$
\end{itemize}
\end{multicols}

with $\i_0=55$, $\i_1=54$, $i_2=53$ and $\p=144$.
We call $\IL(0)$, $\IL(1)$ and $\IL(2)$ left \emph{information} gadgets and $\IR(0)$, $\IR(1)$ and $\IR(2)$ right information gadgets. 
We say that $\IL(\alpha)$ and $\IR(\beta)$ \emph{match} if $\alpha=\beta$ or \emph{mismatch} otherwise.
$\PL$ and $\PR$ are called left and right \emph{protector} gadgets, respectively.
$\SL$ and $\SR$ are called  left and right \emph{synchronizer} gadgets, respectively.
We say that two gadgets $g_1,g_2$ are \emph{symmetric} if $g_2=\rev{g_1}$. 
Specifically, $(\IL(0),\IR(0))$, $(\IL(1),\IR(1))$, $(\IL(2), \IR(2))$ and $(\PL,\PR)$ are the pairs of symmetric gadgets.
We emphasize that $\SL$ and $\SR$ are \textit{not} symmetric.

Using the gadgets above, we define $6$ mega gadgets encoding characters from $S$ and $T$.
For $\alpha\in\{0,1,2\}$ we define 
\[
\EL(\alpha)=\PL\cdot\SL\cdot\IL(\alpha)\cdot\SL
\text{ \hspace{10px} and  \hspace{10px} }
\ER(\alpha)=\SR\cdot \IR(\alpha)\cdot \SR\cdot \PR.
\]
Finally, we define $y=\boxed {\PL\cdot\SL\cdot 01\cdot11\inv{11}\cdot\inv{10}\cdot\SR\cdot\PR}$ and $x=\left(\bigodot_{i=1}^{|S|}\EL(S[i])\right) \cdot y\cdot \left(\bigodot_{i=|T|}^{1}\ER(T[i])\right)$. 
Note that on the suffix of $x$ we concatenate the elements of $T$ in \emph{reverse} order.

In the remainder of this paper, we only use '$x$' and '$y$' to refer to the strings defined above. 
We define notations for indices in $x$ which are endpoints of protector and information gadgets as follows.
For $\ell\in[|S|+1]$, let $\leftP_\ell=1+\sum_{j=1}^{\ell-1}|\EL(S[j])|$ be the leftmost index of the $\ell$th $\PL$ gadget (from the left) in $x$.
Notice that $\leftP_{|S|+1}$ corresponds to the left $\PL$ gadget contained in $y$.
For $r\in[|T|+1]$, let $\rightP_r=|x|-\sum_{j=1}^{r-1}|\ER(T[j])|$ be the rightmost index of the $r$th $\PR$ gadget (from the right) in $x$.
Notice that $\rightP_{|T|+1}$ corresponds to the right $\PR$ gadget contained in $y$.
For $\ell\in[|S|]$, let $\leftI_\ell=\leftP_\ell+|\PL|+|\SL|$ be the leftmost index of the $\ell$th information gadget (from the left) in $x$.
For $r\in[|T|]$, let $\rightI_r=\rightP_r-|\PR|-|\SR|$ be the rightmost index of the $r$th information gadget (from the right) in $x$.

The rest of the paper is dedicated for proving the following property of $x$ and $y$.
\begin{restatable}[Reduction Correctness]{lemma}{lemreduction}

\label{lem:reduction}
For some constants $D(0),D(1),D(2)$ and $B$ we have:
    $\HDD(x,y) =\sum_{\alpha\in\{0,1,2\}} D(\alpha)(\#_\alpha(S)+\#_\alpha(T)) - \LCS(S,T) \cdot B$.
\end{restatable}
Note that $\#_\alpha(S)$ and $\#_\alpha(T)$ can be easily computed for all values of $\alpha$ in linear time.
Therefore, if $\HDD(x,y)$ can be computed in $O(n^{2-\eps})$ time for some $\eps > 0$, $\LCS$ can be computed in $O(n+ n^{2-\eps})$ (the values of the constants are fixed in the proof).
Since $\HDD(x,y) = \HC(y,x)$, hairpin deletion and hairpin completion distance are computationally equivalent.
Recall that $\HDD(x,y)$ refers to the \textit{modified} hairpin deletion distance (\cref{def:mod_hairpin}).
In order to bridge the gap to the original definition of hairpin deletion distance, we provide a linear time construction of strings $x'$ and $y'$ such that $\HDD'(x',y')=\HDD(x,y)$ in \cref{app:equivalence}. 
Here, $\HDD'(x',y')$ denotes the hairpin deletion distance from $x'$ to $y'$ as defined in \cref{def:hairpin}.
It clearly follows from this construction and the above discussion that $\HDD'(x',y')$ can not be computed in $O(n^{2-\eps})$, unless SETH is false.
Thus, proving \cref{thm:reduction}.

\subsection{Intuition for the Reduction Correctness}\label{sec:overview}
We provide some high-level discussion regarding the correctness of the construction.
First, notice that $y$ has a single occurrence in $x$.
Therefore, a sequence of hairpin deletion operations transforming $x$ into $y$ has to delete all mega gadgets.
Consider an intermediate step in a deletion sequence in which the substring $x[i..j]$ is obtained such that $i$ is the leftmost index of some left gadget $g_i$ and $j$ is the rightmost index of some right gadget $g_j$.

If $g_i$ and $g_j$ are not symmetric, the next hairpin deletion would not be able to make much progress.
This is due to the $1$ symbols in $g_i$ and the $\overline{1}$ symbols in $g_j$ being separated by a different number of $0$'s and $\overline{0}$'s.
For the goal of minimizing the number of deletions for removing all mega gadgets, this is a significant set-back, as either $g_i$ or $g_j$ would have to be removed using roughly $\#_1(g_i)$ (or $\#_{\overline{1}}(g_j)$) deletions.

Now consider the case in which $g_i$ and $g_j$ are symmetric to each other.
In this case, either one of them can be deleted using a single hairpin deletion. 
However, note that greedily removing $g_i$ will put us in the asymmetric scenario.
Notice that there is another possible approach for deleting symmetric gadgets - a synchronized deletion.
In this process, $g_i$ and $g_j$ are both deleted gradually.
One can easily figure out a way to apply such synchronized deletion using roughly $\log (\#_1(g_i))$ steps.

Think of a scenario in which $i=\leftP_\ell$ and $j=\rightP_r$ for some $\ell$ and $r$, i.e. $i$ and $j$ are a leftmost index and a rightmost index of left and right mega gadgets $m_\ell$ and $m_r$, respectively.
Initially, both $i$ and $j$ are in the beginning of protector gadgets $p_\ell$ and $p_r$.
If $m_\ell$ and $m_r$ are mega gadgets corresponding to the same symbol $\alpha$, the protector gadgets and the information gadgets of $m_\ell$ and $m_r$ are symmetric to each other. 
It is therefore very beneficial to remove the mega gadgets in a synchronized manner.
The event in which the $\ell$'th left mega gadget and the $r$'th right mega gadgets are deleted in a synchronized manner corresponds to $S[\ell]$ and $T[r]$ being matched by the longest common subsequence.

Now, consider the case in which $m_\ell$ and $m_r$ do not match i.e. $S[\ell] \neq T[r]$.
Deleting the protectors in a synchronized manner would not yield much benefit in this scenario, since the information gadgets are not symmetric, and therefore would have to be removed slowly.
In this case, since an inefficient deletion of an information gadget is inevitable, it is more efficient to delete one of the protectors gadgets using a single deletion operation, and proceed to delete the following information gadget inefficiently.
This would result in either the left mega gadget being deleted, or the right one.
The event in which the $\ell$'th left (resp. $r$'th right) mega gadget is deleted in a non synchronized manner corresponds to $S[\ell]$ (resp. $T[r]$) not being in the longest common subsequence.
The gadgets are designed in a way such that deleting mega gadgets in a synchronized way is faster than deleting each mega gadget in a non-synchronized way.
Furthermore, the cost reduction of a synchronized deletion over a non-synchronized deletion is a constant number $B$.
Therefore, by selecting $D(\alpha)$ as the cost of deleting a mega gadget corresponding to the symbol $\alpha$ in a non-synchronized way, one obtains \cref{lem:reduction}.

The above discussion makes an implicit assumption that the sequence of deletion is applied in \textit{phases}.
Each phase starts with $x[i..j]$ such that $i$ and $j$ are edge endpoints of mega gadgets and proceeds to either delete both in a synchronized manner or one in a non-synchronized manner.
In order to show that $\HDD(x,y)$ is at most the term in \cref{lem:reduction}, this is sufficient since we can choose a sequence of deletion with this structure as a witness.
In order to show that $\HDD(x,y)$ is at least the expression in \cref{lem:reduction}, one has to show that there is an optimal sequence of deletions with this structure.
This is one of the main technical challenges in obtaining \cref{thm:reduction}.

In a high level, the $\sync$ gadgets function as `anchors' that force any sequence of deletions to stop in their proximity.
Another key property of our construction that enforces the `phases' structure is the large size of a protector relatively to the information.
Intuitively, an optimal sequence would always avoid deleting a protector gadget inefficiently, so if a left protector is deleted using a right protector, left deletions would continue to occur until the next left protector is reached. 

In \cref{sec:well}, we provide the formal definition for a well-structured sequence and prove that there is an optimal sequence of deletions with this structure.
In \cref{sec:costSteps} we provide a precise analysis of every phase in a well-structured sequence.
In \cref{sec:correctness} we put everything together and prove \cref{lem:reduction}.

\section{Well-Behaved Paths}\label{sec:well}
We start by defining a well-behaved path.
\begin{definition}[Well-Behaved Path]\label{def:wellbehavedpath}
A path $p$ from $x$ to $y$ in $G_x$ is well-behaved if for every $\ell \in [|S|+1]$ and $r \in [|T|+1]$, if $p$ visits $x[\leftP_\ell..\rightP_r]$, one of the following vertices is also visited by
     $p$:
      $x[\leftP_{\ell+1} .. \rightP_{r}]$, $x[\leftP_\ell .. \rightP_{r+1}]$, or $x[\leftP_{\ell+1} .. \rightP_{r+1}]$.
      If one of $\ell+1$ and $r+1$ is undefined, the condition is on the subset of defined vertices. 
      If both are undefined, the condition is considered satisfied.
\end{definition}

This section is dedicated to proving the following lemma.
\begin{restatable}[Optimal Well-Behaved Path]{lemma}{lemmaintech}
\label{lem:maintech}
There is a shortest path from $x$ to $y$ in $G_x$ which is well-behaved.
\end{restatable}

We start by proving properties regarding paths and shortest paths from $x$ to $y$ in $G_x$.
Due to space constraints, the proofs for the lemmata in this section appear in \cref{app:missingwell}.

\subsection{Properties of Paths in \texorpdfstring{$G_x$}{Gₓ}}\label{sec:local}
We start by observing that an $x$ to $y$ path in $G_x$ never deletes symbols from $y$.
\begin{restatable}[Never Delete $y$]{observation}{nodelyObs}
\label{obs:nodely}
    The substring $x[\leftP_{|S|+1} .. \rightP_{|T|+1}]=y$ is the unique occurrence of $y$ in $x$.
    Let $p$ be a path from $x$ to $y$ in $G_x$. 
    For every vertex $x[i..j]$ in $p$, $[\leftP_{|S|+1} .. \rightP_{|T|+1}] \subseteq [i..j]$. 
\end{restatable}

Since each hairpin deletion operation deletes a prefix (or a suffix) of a substring of $x$, we have the following observation and immediate corollary.
\begin{observation}\label{obs:diffperiod}
    Let $x[i..j]\in 010^a1 * \inv{10^b10}$ for $a\ne b$. 
    A single left hairpin deletion operation removes at most a single $1$ character.
    Symmetrically, a right hairpin deletion operation removes at most a single $\inv{1}$ character.
\end{observation}

\begin{corollary}\label{clm:single1char}
    A single left (resp. right) hairpin deletion operation on $x[i..j]$ can remove more than a single $1$ (resp. $\inv{1}$) character only if $i$ and $j$ are in symmetric gadgets.
\end{corollary}

The next lemma assures a restriction over the vertices along a path from $x$ to $y$ in $G_x$.
\begin{restatable}[Always $01 *$ or $*\inv{10}$]{lemma}{lemzerooneandonezero}
\label{lem:01and10}
     Let $p$ be a path from $s$ to $t$ in $G_x$ such that $s,t \in 01 * \inv{10}$. 
     For every vertex $x[i..j]$ visited by $p$, we have $x[i]=0$ and $x[j] = \inv0$.
     Furthermore, $x[i+1] = 1$ or $x[j-1] = \inv1$.
\end{restatable}

Due to the equivalence between a path in $G_x$ and a sequence of hairpin deletions and due to the symmetry between hairpin deletion and hairpin completion, we obtain the following.
\begin{corollary}\label{cor:0110deletion}
    Let $s,t \in 01 * \inv{10}$ and
    let $\mathcal{H}$ be a sequence of $h$ hairpin deletion operations (or a sequence of hairpin completion operations) that transforms $s$ into $t$.
    For $i\in [h]$, let $S_i$ be the string obtained by applying the first $i$ operations of $\mathcal{H}$ on $s$.
     For every $i \in [h]$, we have $S_i[1] = 0 $ and $S_i[|S_i|] = \inv{0}$.
     Furthermore, either $S_i[2] = 1$ or $S_i[|S_i| - 1] = \inv{1}$.
\end{corollary}

The following lemma discusses the situation in which $p$ visits a vertex not in $01 * \overline{10}$. 
Essentially, the lemma claims that when $p$ visits a substring $x[i..j]$ with a prefix $00$, the next step would be $x[i+1..j]$, i.e., deleting a single zero from the left. 
\begin{restatable}[Return to $01*10$]{lemma}{lemretzerooneandonezero}
\label{lem:ret0110}
     Let $p$ be a path from $x$ to $y$ in $G_x$. 
     If $p$ visits a vertex $x[i..j]$ such that $x[i..j] = 01 * \inv{10^k}$ for some integer $k \ge 1$, then for every $k'\in[k-1]$ it must be that $p$ visits $x[i.. j - k']$ as well.
     Symmetrically, if $p$ visits a vertex $x[i..j]$ such that $x[i..j] = 0^k1 * \inv{10}$ for some integer $k \ge 1$, then
     for every $k'\in[k-1]$ it must be that $p$ visits $x[i+k'.. j]$ as well.
\end{restatable}

The following is a direct implication of \cref{lem:01and10,lem:ret0110}.
\begin{observation}\label{obs:stableside01}
Let $p$ be a path from $x$ to $y$ in $G_x$. 
If $p$ applies a right hairpin deletion operation on $v$ then $v\in01*$.
Symmetrically, if $p$ applies a left hairpin deletion operation on $v$ then $v\in*\inv{10}$.
\end{observation}

The following lemma establishes the importance of the synchronizer gadgets.
\begin{restatable}[Synchronizer Lemma]{lemma}{lemsync}
\label{lem:sync}
    Let $p$ be a path from $x$ to $y$ in $G_x$ and let $s = x[i_s..j_s] = \SL$ be a left synchronizer which is not contained in $y$, $p$ must visit a vertex $x[j_s+1..k]$ for some integer $k$.
    Symmetrically, if $x[i_s..j_s] = \SR$ is a right synchronizer which is not contained in $y$, $p$ must visit a vertex $x[k .. i_s-1]$.
\end{restatable}

Notice that the leftmost index of every left information and protector gadget is $j_s+1$ for some left synchronizer $x[i_s..j_s]$ (excluding the leftmost protector gadget).
A similar structure occurs with right gadgets.
The following directly follows from \cref{lem:sync}.
\begin{corollary}\label{cor:sync}
    Let $p$ be a path from $x$ to $y$ in $G_x$.
    Then, for each $\ell\in [|S|]$ the path $p$ visits vertices $u=x[i.. j]$ with $i=\leftP_\ell$ and $v=x[i.. j]$ with $i=\leftI_\ell$.
    Symmetrically, for each $r\in [|T|]$ the path $p$ visits vertices $u=x[i.. j]$ with $j=\rightP_r$ and $v=x[i.. j]$ with $j=\rightI_r$.
\end{corollary}

\subsection{Transitions Between Gadegets}\label{sec:trans}
In this section, we address the way shortest paths apply to vertices that transit from a gadget to the gadget afterward.

\begin{restatable}{lemma}{lemidnotseer}
\label{lem:idontseer}
    Let $p$ be a shortest path from $x$ to $y$ in $G_x$.
    For some $\ell \in |S|$, let $v=x[i_1..j_1]$ be the first vertex visited by $p$ with $i_1=\leftI_\ell$.
    Let $u=x[i_2..j_2]$ be the first vertex visited by $p$ with $i_2 = \leftP_{\ell+1}$.
    Then, there is no occurrence of $\PR$ in $x[j_2..j_1]$.

    Symmetrically, for some $r \in |T|$ let $v=x[i_1..j_1]$ be the first vertex visited by $p$ with $j_1=\rightI_r$.
    Let $u=x[i_2..j_2]$ be the first vertex visited by $p$ with $j_2 = \rightP_{r+1}$.
    Then, there is no occurrence of $\PL$ in $x[i_1..i_2]$.
\end{restatable}
\begin{proof}[Proof Sketch, Complete Proof in \cref{app:missingtrans}]
The proof is by contradiction. 
If there is only a \emph{single} $\PR$ gadget that is contained in $x[j_2..j_1]$, by \cref{clm:single1char} the number of hairpin deletions that must happen just to remove this gadget is at least $\p$.
We introduce an alternative, shorter path:
At the moment $p$ reaches this $\PR$ gadget, it first removes all the $\IL(S[\ell])$ gadget, and then removes all the remaining characters on the right side greedily, until reaching $x[i_2..j_2]$.
The reason why this alternative path is indeed shorter is since in this way the removal of the $\PR$ gadget takes $1$ operation, instead of $\p$, and we may pay at most $\i_{\alpha}+\i_{\beta}\le 2\i_0$ for some $\alpha,\beta\in\{0,1,2\}$, for removing one information gadget in the left side and the information gadget following the right protector gadget.
Since $\p$ is much larger than $\i_0$, the alternative path is shorter, contradicting the assumption that $p$ is a shortest path.
Notice that if there are more than one $\PR$ gadgets in $x[j_2..j_1]$, the benefit from deleting $\IL(S[\ell])$ first is even larger.
\end{proof}

The following lemma states that every right deletion on $x[i..j]$ with $i$ being within a non $\SL$ gadget can also be applied if $i$ is the leftmost index of the gadget. 
A symmetric argument is stated as well.

\begin{restatable}{lemma}{lemprotDelay}
\label{lem:protDelay}
    Let $p$ be a path from $x$ to $y$ in $G_x$.
    Let $v=x[i..j]$ be a vertex visited by $p$ such that $i\in[\leftP_\ell..\leftI_\ell-1]$ for some $\ell\in[|S|]$.
    Let $v'=x[\leftP_\ell..j]$, let $u=x[i..j-k]$ and $u'=x[\leftP_\ell..j-k]$ for some $k$.
    If $(v,u)$ is an edge in $p$, then $(v',u')$ is an edge in $G_x$.

    The above statement considers the case in which $v$ interacts with a $\PL$ gadget, the following similar statements, regarding different gadgets hold as well:
    \begin{itemize}
    \item \textbf{$\PR$:} Let $v=x[i..j]$ be a vertex visited by $p$ such that $j\in[\rightI_r+1..\rightP_r]$ for some $r\in[|T|]$.
    Let $v'=x[i..\rightP_r]$, let  $u=x[i+k..j]$ for some $k$, and let $u'=x[i+k..\rightP_r]$
    If $(v,u)$ is an edge in $p$, then $(v',u')$ is an edge in $G_x$.

    \item \textbf{$\IL(\alpha)$ for some $\alpha\in\{0,1,2\}$:} Let $v=x[i..j]$ be a vertex visited by $p$ such that $i\in[\leftI_\ell..\leftP_{\ell+1}-1]$ for some $\ell\in[|S|]$.
    Let $v'=x[\leftI_\ell..j]$, let $u=x[i..j-k]$ and $u'=x[\leftI_\ell..j-k]$ for some $k$.
    If $(v,u)$ is an edge in $p$, then $(v',u')$ is an edge in $G_x$.

    \item \textbf{$\IR(\alpha)$ for some $\alpha\in\{0,1,2\}$:} 
    Let $v=x[i..j]$ be a vertex visited by $p$ such that $j\in[\rightP_{r+1}+1..\rightI_r]$ for some $r\in[|T|]$.
    Let $v'=x[ i..\rightI_r]$, let  $u=x[i+k..j]$ for some $k$, and let $u'=x[i+k..\rightI_r]$
    If $(v,u)$ is an edge in $p$, then $(v',u')$ is an edge in $G_x$.
\end{itemize}

\end{restatable}
\begin{proof}[Proof Sketch, Complete Proof in \cref{app:missingtrans}]
We distinguish between two cases. 
If $x[i..i+1]\ne01$, by \cref{lem:ret0110} the hairpin deletion removes exactly one $\inv{0}$ character, and by $x[\leftP_\ell]=0$ the edge $(u',v')$ is in $G_x$.
If $x[i..i+1]=01$, we first prove (using \cref{cor:sync}) that $k\le \leftI_\ell-i$.
Moreover, since $i\in [\leftP_\ell..\leftI_\ell-1]$ it must be the case that $i=\leftP_\ell+q\cdot 11$ for some $q$.
Since $x[\leftP_\ell..\leftI_\ell-1]$ is periodic with period $11$.
Thus, $x[\leftP_\ell..\leftP_\ell+k]=x[\leftP_\ell+q\cdot11..\leftP_\ell+q\cdot11+k]=x[i..i+k]$ and the claim follows.  
\end{proof}

We are now ready to prove \cref{lem:maintech}.
\lemmaintech*
\begin{proof}

We describe a method that converts a shortest path $p$ from $x$ to $y$ that visits $u=x[\leftP_\ell..\rightP_r]$ into a shortest path $p'$ from $x$ to $y$ that visits one of the following vertices: $x[\leftP_{\ell+1} .. \rightP_{r}]$, $x[\leftP_\ell .. \rightP_{r+1}]$, or $x[\leftP_{\ell+1} .. \rightP_{r+1}]$.
Moreover, the prefixes of $p$ and $p'$ from $x$ to $u$ are identical.
Using this technique, it is straightforward to convert a shortest path from $x$ to $y$ in $G_x$ into a well-behaved path of the same length.

Let $v_L=x[i_L..j_L]$ be the first vertex in $p$ with $i_L=\leftP_{\ell+1}$ and let $v_R=x[i_R..j_R]$ be the first vertex in $p$ with $j_R=\rightP_{r+1}$.
By \cref{cor:sync}, $v_L$ and $v_R$ are well defined (unless $\ell=|S|+1$ or $r=|T|+1$, in such a case just one of the vertices is well defined and the claim follows trivially from \cref{obs:nodely}).
We consider the case where $v_L$ appears before $v_R$ in $p$ and show how to convert $p$. 
The other case is symmetric.

We distinguish between two cases:
\para{Case 1 $j_L\in[\rightI_{r}+1..\rightP_{r}]$:}
Let $q$ be the sub-path of $p$ from $u$ to $v_L$.
We present a path $q^*$ from $u$ to $v_L$ that is not longer than $q$ and visits $x[\leftP_{\ell+1}..\rightP_{r}]$.
Recall that an edge of the form $x[i..j]\to x[i+k..j]$ is called a \emph{left} edge, and an edge of the form $x[i..j]\to x[i..j-k]$ is called a \emph{right} edge.
Let $\cost_L$ be the number of left edges in $q$ and $\cost_R$ be the number of right edges in $q$. 
We first show a path from $u=x[\leftP_\ell..\rightP_r]$ to $x[\leftP_{\ell+1}..\rightP_r]$ of length $\cost_L$.
Let $e=x[i_1..j]\to x[i_2..j]$ be a left edge in $q$.
It must be that $j\in[j_L..\rightP_r]\subseteq [\rightI_r+1..\rightP_r]$.
Hence, by \cref{lem:protDelay}, there exists an edge $e'=x[i_1..\rightP_r]\to x[i_2..\rightP_r]$.
Let $e_1,e_2,\ldots,e_{\cost_L}$ be the subsequence of all left edges in $q$.
The path $q^*_1=e'_1,e'_2,\ldots,e'_{\cost_L}$ is a valid path of length $\cost_L$ from $u=x[\leftP_\ell..\rightP_r]$ to $x[\leftP_{\ell+1}..\rightP_r]$.

If $j_L=\rightP_r$ then $q^*=q^*_1$ is a path that satisfies all the requirements. 
Otherwise, $\cost_R\ge 1$.
We claim that there is an edge $e_R$ from $x[\leftP_{\ell+1}..\rightP_r]$ to $v_L=x[\leftP_{\ell+1}..j_L]$.
This is true since  $x[\leftP_{\ell+1}..\leftI_{\ell+1}-1]=\PL\cdot\SL=\rev{\SR[2..|\SR|]\cdot\PR}=\rev{x[\rightI_{r}+2\,\,..\,\,\rightP_r]}$ and $j_L\in[\rightI_{r}+1..\rightP_{r}]$.
We conclude $q^*$ by appending $e_R$ to the end of $q^*_1$. 
Finally, $\cost(q^*)=\cost_L+1\le \cost_L+\cost_R=\cost(q)$, and $q^*$ visits $x[\leftP_{\ell+1}..\rightP_r]$.

\para{Case 2 $j_L\in[\rightP_{r+1}-1..\rightI_{r}]$:}
We first prove the following claim.
\begin{claim*}
$i_R\in [\leftP_{\ell+1}..\leftI_{\ell+1}-1]$.
\end{claim*}
\begin{claimproof}
Since $v_R$ appears after $v_L$, we have $i_R \ge i_L = \leftP_{\ell+1}$.
Assume to the contrary that $i_R \ge \leftI_{\ell+1}$.
Let $v_f=x[i_f..j_f]$ be the first vertex in $p$ with $j_f = \rightI_r$ ($v_f$ exists according to \cref{cor:sync}).
Note that $v_f$ does not appear after $v_L$ in $p$ since $j_L=\le \rightI_{r}= j_f$
Therefore, $i_f\le i_L=\leftP_{\ell+1}$ and $[\leftP_{\ell+1}..\leftI_{\ell+1}]\subseteq[i_f..i_R]$.
Therefore, the occurrence of $\PL$ starting in $\leftP_{\ell+1}$ is contained in $x[i_f.. i_R]$.
Since $p$ is a shortest path, this is a contradiction to \cref{lem:idontseer}.
\end{claimproof}

Let $q$ be the sub-path of $p$ from $v_L$ to $v_R$.
Let $\cost_L$ be the number of left edges in $q$ and $\cost_R$ be the number of right edges in $q$. 
We present a path $q^*$ from $v_L$ to $v_R$ that is not longer than $q$ and visits $x[\leftP_{\ell+1}..\rightP_{r+1}]$.
We first show a path $q^*_1$ from $v_L=x[\leftP_{\ell+1}..j_L]$  to $x[\leftP_{\ell+1}..\rightP_{r+1}]$  of length $\cost_R$.
Let $e=x[i..j_1]\to x[i..j_2]$ be a right edge in $q$.
It must be that $i\in[\leftP_{\ell+1}..i_R]\subseteq [\leftP_{\ell+1}..\leftI_{\ell+1} - 1]$ due to the claim.
Hence, by \cref{lem:protDelay}, there exists an edge $e'=x[\leftP_{\ell+1}..j_1]\to x[\leftP_{\ell+1}..j_2]$.
Let $e_1,e_2,\ldots,e_{\cost_L}$ be the subsequence of all right edges in $q$.
The path $q^*_1=e'_1,e'_2,\ldots,e'_{\cost_L}$ is a valid path of length $\cost_R$ from $v_L=x[\leftP_{\ell+1}..j_L]$ to $x[\leftP_{\ell+1}..\rightP_{r+1}]$.

If $i_R=\leftP_{\ell+1}$ then $q^*=q^*_1$ is a path that satisfies all the requirements. 
Otherwise, $\cost_L\ge 1$.
We claim that there is an edge $e_L$ from $x[\leftP_{\ell+1}..\rightP_{r+1}]$ to $v_R=x[i_R..\rightP_{r+1}]$.
This is true since  $x[\leftP_{\ell+1}..\leftI_{\ell+1}]=\PL\cdot\SL \cdot 0=\rev{\SR\cdot\PR}=\rev{x[\rightI_{r+1}+1\,\,..\,\,\rightP_{r+1}]}$ and $i_R\in[\leftP_{\ell+1}..\leftI_{\ell+1} - 1]$.
We conclude $q^*$ by appending $e_L$ to the end of $q^*_1$. 
Finally, $\cost(q^*)=\cost_L+1\le \cost_L+\cost_R=\cost(q)$, and $q^*$ visits $x[\leftP_{\ell+1}..\rightP_r]$.
\end{proof}

\section{Cost of Well-Behaved Steps}\label{sec:costSteps}
In this section, we analyze the cost of each of the possible phases of a well-behaved path (\cref{def:wellbehavedpath}).
We first consider the cost of deletion of a single mega-gadget.
\begin{lemma}[Non Synchronized Deletion]\label{lem:nonsyncdel}
Let $v = x[\leftP_{\ell} .. \rightP_{r}]$, $u_1 = [\leftP_{\ell+1} .. \rightP_{r}]$ and $u_2 = [\leftP_{\ell} .. \rightP_{r+1}]$ for $\ell \in [|S|]$ and $r \in [|T|]$.
Let $S[\ell] = \alpha$ and $T[r] = \beta$. 
It holds that $\dist_{G_x}(v,u_1) = \i_{\alpha} + 2$ and $\dist_{G_x}(v,u_2) = \i_{\beta} + 2$.    
\end{lemma}
\begin{proof}
We prove $\dist_{G_x}(v,u_1) = \i_{\alpha} + 2$. The proof for $\dist_{G_x}(v,u_2) = \i_{\beta} + 2$ is symmetrical.
We prove the lemma by showing $\dist_{G_x}(v,u_1) \ge \i_{\alpha} + 2$ and $\dist_{G_x}(v,u_1) \le \i_{\alpha} + 2$.
\para{$\dist_{G_x}(v,u_1) \ge \i_{\alpha} + 2$:}
Let $p$ be a $v$ to $u_1$ path in $G_x$.
Note that vertex $x[i..j]$ in $p$ has $j= \rightP_r$.
According to \cref{cor:sync}, $p$ must visit $z = x[\leftI_{\ell} .. \rightP_{r}]$. 
Since $z \neq v$, the sub-path of $p$ from $v$ to $z$ induces a cost of at least $1$ to $p$.
Consider the sub-path $q$ of $p$ from $z$ to $u_1$.
According to \cref{clm:single1char}, every left hairpin deletion step in $q$ deletes at most a single '1' symbol.
Due to $x[\leftI_{\ell} .. \leftP_{\ell+1} - 1] = \IL(\alpha) \cdot \SL$, the sub-path $q$ consists of at least $\#_1(\IL(\alpha)) + 1 = \i_\alpha + 1$ additional left hairpin deletions.

\para{$\dist_{G_x}(v,u_1) \le \i_{\alpha} + 2$:}
We present a path $p$ with cost exactly $\i_\alpha+2$ from $v$ to $u_1$.
Initially, $p$ deletes a prefix of length $|\PL| + |\SL|$ from $v$ in one step. 
This is possible since $v$ has a suffix $\inv{10} \cdot \PR$.
Then, $p$ proceeds to delete $x[\leftI_{\ell} .. \leftP_{\ell+1} - 1] = \IL(\alpha) \cdot \SL$ a single '$1$' character at a time.
Note that this is possible regardless of the value of $\alpha$ due to $x[\rightP_r - 8 .. \rightP_r] = \inv{0^710}$.
The total cost of this path is $\i_\alpha + 2 $ as required.
\end{proof}

In the following lemma, we show that the cost of deleting two disagreeing mega-gadgets is the same as deleting each one of them separately.

\begin{lemma}[Synchronized Deletion of Disagreeing Mega Gadgets]\label{cor:syncdeldisagree}
Let $v = x[\leftP_{\ell} .. \rightP_{r}]$, $u =[\leftP_{\ell+1} .. \rightP_{r+1}]$  for $\ell \in [|S|]$ and $r \in [|T|]$ with $S[\ell]\ne T[r]$.
It holds that $\dist_{G_x}(v,u) = \i_\alpha +\i_\beta + 4$ with $S[\ell] = \alpha$ and $T[r] = \beta$.
\end{lemma}
\begin{proof}
We prove the claim by showing $\dist_{G_x}(v,u) \ge \i_\alpha +\i_\beta + 4$ and $\dist_{G_x}(v,u) \le \i_\alpha +\i_\beta + 4$.

\para{\underline{$\dist_{G_x}(v,u) \ge \i_\alpha +\i_\beta + 4$:}}
Let $p$ be a path from $v$ to $u$ in $G_x$.
According to \cref{cor:sync}, $p$ visits vertices $z_1 = x[\leftI_{\ell} .. j]$ and $z_2 = x[i .. \rightI_{r}]$ for some $i,j$. 
The last left hairpin deletion in $p$ before $z_1$ and the last right hairpin deletion in $p$ before $z_2$ induce a cost of $2$ to $p$.
Consider a left hairpin deletion that is applied to a vertex $x[i'..j']$ after $z_1$ in $p$.
Note that $i'$ is either within an $\IL(\alpha)$ gadget or within a $\SL$ gadget, and $j'$ is either within an $\IR(\beta)$ gadget, a $\SR$ gadget or a $\PR$ gadget.
In any of the above cases, \cref{clm:single1char} suggests that the deletion operation deletes at most a single '1' character. 
Therefore, there are at least $\i_\alpha + 1$ left deletions after $z_1$ in $p$.
Due to similar reasoning, there are at least $\i_\beta + 1$ right hairpin deletions after $z_2$ in $p$.
It follows that the total cost of $p$ is at least $\i_\alpha + \i_\beta + 4$.

\para{\underline{$\dist_{G_x}(v,u) \le \i_\alpha +\i_\beta + 4$:}}
Consider the path $p$ that is composed of two sub-paths, the prefix $p_1$ is a shortest path from $v$ to $w=x[\leftP_{\ell+1}..\rightP_r]$ and the suffix $p_2$ is a shortest path from $w$ to $u$.
By \cref{lem:nonsyncdel} we have $\cost(p_1)=\i_{\alpha}+2$ and $\cost(p_2)=\i_{\beta}+2$. 
Therefore $\cost(p)=\cost(p_1)+\cost(p_2)=\i_\alpha+2+\i_\beta+2$. \end{proof}

The last case we have to analyze is a synchronized deletion of agreeing mega gadgets.
We first present the concept of \textit{Fibonacci-regular numbers}.
\begin{definition}[Fibonacci-regular number]
    We say that $a\in\N$ is a \emph{Fibonacci-regular} number if for all $2\le k\le a$ it holds that  $\Fib^{-1}(a)\le\Fib^{-1}(a/k)+k-1$.
\end{definition}
\begin{fact}\label{fact:fibReg}
    $\i_2=53$, $\i_1=54$, $\i_0=55$ and $\p=144$ are Fibonacci-regular numbers.
\end{fact}

The following lemma, which \cref{app:FibDel} is dedicated for the lemma's proof, provides the required machinery to analyze the cost of a synchronized deletion.
\begin{restatable}{lemma}{FibDistanceLemma}
\label{lem:FibDistance}
Let $\per=010^{\ext}$ and let $q\in010^{\int}01*11\inv{11}*\inv{100^{\int}10}$ with $\int \ne \ext$ and $\min\{\int,\ext \}\ge3$.
For every Fibonacci-regular number $a$, we have  
$\HDD(\per^a\cdot\SL\cdot  q\cdot \SR \cdot \rev{\per}^a,q)=\Fib^{-1}(a)+ \max(\ext - \int - 1,0 ) +3$
\end{restatable}

Finally, we are ready to analyze the cost of synchronized deletion of agreeing mega gadgets.
\begin{lemma}[Synchronized Deletion of Agreeing Mega Gadgets]\label{lem:syncdelagree}
Let $v = x[\leftP_{\ell} .. \rightP_{r}]$, $u = [\leftP_{\ell+1} .. \rightP_{r+1}]$  for $\ell \in [|S|]$ and $r \in [|T|]$ with $S[\ell]= T[r]=\alpha$.
Then  $\dist_{G_x}(v,u)=\Fib^{-1}(\p)+\Fib^{-1}(\i_\alpha)+11-2\alpha$.
\end{lemma}
\begin{proof}

Let $w=x[\leftI_{\ell} .. \rightI_{r}]$.
Consider the following path $p'$ from $u$ to $v$ in $G_x$.
The path $p'$ consists of a prefix $p'_1$ which is a shortest path from $u$ to $w$ and a suffix $p'_2$ which is a shortest path from $w$ to $v$.
Since $\i_0,\i_1,\i_2$ and $\p$ are Fibonacci-regular numbers (\cref{fact:fibReg}), according to \cref{lem:FibDistance} (with $\ext=9$ and $\int=3+2\alpha$) we have $\cost(p'_1)=\Fib^{-1}(\p)+\max(9-(3+2\alpha)-1,0)+3=\Fib^{-1}(\p)+8-2\alpha$.
Similarly, according to \cref{lem:FibDistance} (with $\ext=3+2\alpha$ and $\int=9$) we have  $\cost(p'_2)=\Fib^{-1}(\i_\alpha)+\max(3+2\alpha-9-1,0)+3=\Fib^{-1}(\i_\alpha)+3$.
In total we have $\cost(p')=\cost(p'_1)+\cost(p'_2)=\Fib^{-1}(\p)+\Fib^{-1}(\i_\alpha)+11-2\alpha$.
Therefore $\dist_{G_x}(v,u)\le\Fib^{-1}(\p)+\Fib^{-1}(\i_\alpha)+11-2\alpha=31 - 2\alpha$.

We prove the following claim.
\begin{claim*}
    There is a shortest path $p$ from $v$ to $u$ that visits $w$.
\end{claim*}
\begin{claimproof}
Let $v_L= x[i_L .. j_L]$ and $v_R = x[i_R ..j_R]$ be the first vertices visited by $p$ with $i_L = \leftI_{\ell}$ and $j_R = \rightI_r$.
Assume without loss of generality that $v_R$ occurs before $v_L$ in $p$.
We consider two cases regarding $j_L$.
\para{Case 1: $j_L = \rightP_{r+1}$}
Consider the suffix $p_s$ of $p$ from $v_L$ to $u$.
Let $v' = x[i'..j']$ be a vertex in $p_s$ that is immediately followed by a left hairpin deletion operation in $p_s$.
Since $i' \in [\leftI_{\ell}.. \leftP_{\ell+1} -1]$ is either within an $\IL(\alpha)$ gadget or within a $\SL$ gadget, and $j' = \rightP_{r+1}$ is in a $\PR$ gadget, \cref{clm:single1char} suggests that the left hairpin deletion applied to $v'$ deletes at most a single '1' character.
It follows from the above analysis that the number of left hairpin deletions in $p_s$ is at least $\#_1(\IL(\alpha)) +1 \ge \i_2 + 1 = 54$.
Therefore, the cost of $p$ is at least $55 >31\ge \cost(p')$, which contradicts the minimality of $p$.

\para{Case 2: $j_L> \rightP_{r+1}$}
Let $q$ be the sub-path of $p$ from $v_R$ to $v_L$.
Let $\cost_L$ be the number of left edges in $q$ and $\cost_R$ be the number of right edges in $q$. 
We first show a path from $v_R$ to $x[\leftI_{\ell}..\rightI_r]$ of length $\cost_L$.
Let $e=x[i_1..j]\to x[i_2..j]$ be a left edge in $q$.
It must be that $j\in[j_L..\rightI_r]\subseteq [\rightP_{r+1} + 1..\rightI_r]$.
Hence, by \cref{lem:protDelay}, there exists an edge $e'=x[i_1..\rightI_r]\to x[i_2..\rightI_r]$.
Let $e_1,e_2,\ldots,e_{\cost_L}$ be the subsequence of all left edges in $q$.
The path $q^*_1=e'_1,e'_2,\ldots,e'_{\cost_L}$ is a valid path of length $\cost_L$ from $v_R$ to $x[\leftI_{\ell}..\rightI_r]$.

If $j_L=\rightI_r$ then $q^*=q^*_1$ is a path that satisfies all the requirements. 
Otherwise, $\cost_R\ge 1$.
We claim that there is an edge $e_R$ from $x[\leftI_{\ell}..\rightI_r]$ to $v_L=x[\leftI_\ell..j_L]$.
This is true since  $x[\leftI_{\ell}..\leftP_{\ell+1}-1]=\IL(S[\ell])\cdot\SL=\IL(T[r])\cdot\SL=\rev{\SR[2..|\SR|]\cdot\IR(T[r])}=\rev{x[\rightP_{r+1}+2\,\,..\,\,\rightI_r]}$ and $j_L\in[\rightP_{r+1}+1..\rightI_{r}]$.
We conclude $q^*$ by appending $e_R$ to the end of $q^*_1$. 
Finally, $\cost(q^*)=\cost_L+1\le \cost_L+\cost_R=\cost(q)$, and $q^*$ visits $x[\leftI_{\ell}..\rightI_r]$.
\end{claimproof}

Let $p$ be a shortest path from $u$ to $v$ in $G_x$.
According to the claim, we can indeed assume that $p$ consists of a shortest path $p_1$ from $v$ to $w$ and a shortest path $p_2$ from $w$ to $v$.
Therefore we have $\cost(p)= \cost(p') =\Fib^{-1}(\p)+\Fib^{-1}(\i_\alpha)+11-2\alpha$ as required.
\end{proof}

\section{Correctness}\label{sec:correctness}
Let $D(0)=57$, $D(1)=56$, $D(2)= 55$, $D_\sync(0)=31$,  $D_\sync(1)=29$, $D_\sync(2) = 27$, and $B = 83$.
The following lemma summarize \cref{lem:nonsyncdel,cor:syncdeldisagree,lem:syncdelagree}.
\begin{lemma}\label{lem:constants}
Let  $\ell\in[|S|]$ and let $r\in[|T|]$ be two integers.
Denote $S[\ell] = \alpha$ and $T[r] = \beta$.
The following is satisfied. 
\begin{enumerate}
        \item $\dist_{G_x}( x[\leftP_\ell ..\rightP_r], x[\leftP_{\ell+1}, \rightP_r]) = D(\alpha)$
        \item $\dist_{G_x}( x[\leftP_\ell ..\rightP_r], x[\leftP_\ell, \rightP_{r+1}]) = D(\beta)$
        \item $\dist_{G_x}( x[\leftP_\ell ..\rightP_r], x[\leftP_{\ell+1}, \rightP_{r+ 1}]) =D(\alpha) + D(\beta)$ if $\alpha \neq \beta$
        \item $\dist_{G_x}( x[\leftP_\ell ..\rightP_r], x[\leftP_{\ell+1}, \rightP_{r+ 1}]) =D_{\sync}(\alpha)$ if $\alpha = \beta$
        \item $2D(0) - D_{\sync}(0) = 2D(1) - D_{\sync}(1) = 2D(2) - D_{\sync}(2) = B$
    \end{enumerate}
\end{lemma}
\begin{proof}
    According to \cref{lem:nonsyncdel}, we have $D(\gamma)=\i_\gamma+2$ for every $\gamma \in \{ 0,1,2\}$. 
    It follows from \cref{cor:syncdeldisagree} that if $\alpha \ne \beta$ we have $\dist_{G_x}( x[\leftP_\ell ..\rightP_r], x[\leftP_{\ell+1}, \rightP_{r+ 1}]) =\i_\alpha + \i_\beta + 4 =\i_\alpha + 2 + \i_\beta+2  = D(\alpha) + D(\beta)$.
    It follows from \cref{lem:syncdelagree} that $D_{\sync}(\gamma) = \Fib^{-1}(\p) + \Fib^{-1}(\i_\alpha) + 11 -2\gamma=11+9+11-2\gamma =31-2\gamma $ for every $\gamma \in \{ 0,1,2 \}$.

    Indeed, we have $2\cdot D(0)  - D_{\sync}(0) = 2 \cdot 57 - 31 = 83$, $2\cdot D(1)- D_{\sync}(1) = 56 \cdot 2 - 29 = 83$  and $2\cdot D(2)- D_{\sync}(2) = 55 \cdot 2 - 27 = 83$ as required.
\end{proof}

We are now ready to prove \cref{lem:reduction} which concludes the correctness of the reduction.

 \lemreduction*
\begin{proof}

We prove the equality claimed, by showing two sides of inequality.
\para{$\HDD(x,y) \le\sum_{\alpha\in\{0,1,2\}} D(\alpha)(\#_\alpha(S)+\#_\alpha(T)) - \LCS(S,T) \cdot B$:}
    Denote $c = \LCS(S,T)$.
    Let $\I=i_1 < i_2 < i_3 .., \ldots , ..< i_c \subseteq[|S|]$ and $\J=  j_1 < j_2 < j_3 \ldots < j_c \subseteq [|T|]$ be two sequences of indices such that $S[i_k] = T[j_k]$ for every $k \in [c]$. 
    Thus, $\I$ and $\J$ represent a maximal common subsequence of $S$ and $T$.

    We present a path $p$ in $G_x$ from $x$ to $y$.
    The path $p$ starts in $x= x[\leftP_1 .. \rightP_1]$, and consists of 3 types of subpaths.
    \begin{enumerate}
        \item Left deletion subpath: a shortest path from $x[\leftP_\ell .. \rightP_r]$ to $x[\leftP_{\ell+1} .. \rightP_{r}]$  for some $\ell\in[|S|]$ and $r\in[|T|+1]$.
          \item Right deletion subpath: a shortest path from $x[\leftP_\ell .. \rightP_r]$ to $x[\leftP_{\ell} .. \rightP_{r+1}]$ for some $\ell\in[|S|+1]$ and $r\in[|T|]$.
          \item Match subpath: a shortest path from $x[\leftP_\ell .. \rightP_{r}]$ to $x[\leftP_{\ell+1} .. \rightP_{r+1}]$ for some $\ell\in[|S|]$ and $r\in[|T|]$.
    \end{enumerate}
    Specifically, if $p$ visits $x[\leftP_\ell .. \rightP_r]$, then $p$ proceeds in a left deletion subpath if $\ell \in [|S|]\setminus\I$.
    Otherwise, $p$ proceeds in a right deletion subpath if $r \in [|T|]\setminus\J$.
    If both $\ell \in \I$ and $r\in \J$, the path $p$ proceeds in a match subpath. 
    Note that it is guaranteed that as long as $\ell \neq |S|+1$ or $ r \neq |T|+1$, $p$ continues to make progress until finally reaching $x[\leftP_{|S|+1} .. \rightP_{|T|+1}]=y$. 
    
We proceed to analyze the cost of $p$.
For $\alpha\in\{0,1,2\}$ we introduce the following notation regarding $\I$ and $\J$.
Let $u_L(\alpha)=|\{i \mid S[i]=\alpha \text{ and } i\notin\I\}|$, $u_R(\alpha)=|\{j\mid T[j]=\alpha\text{ and } j\notin\J\}|$.
In addition, let $c(\alpha)=|\{k\mid k\in[c] \text{ and } S[i_k]=\alpha\}|$.

Clearly, by \cref{lem:constants} every $k\in[c]$ induces a cost of $D_\sync(S[i_k])$ to $p$.
Moreover, every $i\in[|S|]\setminus \I$, induces a cost of $D(S[i])$ to $p$, and every $j\in [|T|]\setminus\J$ induces a cost of $D(T[j])$ to $p$.
Thus, we have
\begin{align*}
\cost(p)&=\sum_{\alpha\in\{0,1,2\}} D(\alpha)(u_L(\alpha)+u_R(\alpha)) + \sum_{\alpha\in\{0,1,2\}} D_\sync(\alpha)\cdot c(\alpha)\\
&=\sum_{\alpha\in\{0,1,2\}} D(\alpha)(u_L(\alpha)+u_R(\alpha)) + \sum_{\alpha\in\{0,1,2\}} (2D(\alpha)-B)\cdot c(\alpha)\\
&=\sum_{\alpha\in\{0,1,2\}} D(\alpha)(u_L(\alpha)+u_R(\alpha)+2c(\alpha)) -B\cdot \sum_{\alpha\in\{0,1,2\}} c(\alpha)\\
&=\sum_{\alpha\in\{0,1,2\}} D(\alpha)(\#_\alpha(S)+\#_\alpha(T)) - c \cdot B.
\end{align*}
Where the first equality follows from $B =  2\cdot D(\alpha) -D_{\sync}$ for every $\alpha\in \{ 0,1,2\}$, and the last equality is since for every  $\alpha\in\{0,1,2\}$ we have $\#_\alpha(S)=u_L(\alpha)+c_\alpha$,$\#_\alpha(T)=u_R(\alpha)+c_\alpha$ and $c=c(0)+c(1)+c(2)$.

\para{$\HDD(x,y) \ge\sum_{\alpha\in\{0,1,2\}} D(\alpha)(\#_\alpha(S)+\#_\alpha(T)) - \LCS(S,T) \cdot B$:}
Let $p$ be a well-behaved shortest path from $x$ to $y$ in $G_x$.
According to \cref{lem:maintech}, such a path $p$ exists.

Let $\X=\{ v = x[\leftP_\ell .. \rightP_r] \mid p \textit{ visits } v \}$.
Notice that the vertices of $\X$ are naturally ordered by the order of their occurrences in $p$, so we denote the $i$th vertex in $\X$ by $x_i$.
For $i\in [|\mathcal{\X}| - 1]$, we classify the vertex $x_i = x[\leftP_\ell .. \rightP_r]$ for some $\ell \in [|S|+1]$ and $r \in [|T| + 1]$ into one of the following four disjoint types.
\begin{enumerate}
    \item \textbf{Match vertex} : if $x_{i+1} = [\leftP_{\ell+1} .. \rightP_{r+1}]$ and $S[\ell]  = T[r] $.
    \item \textbf{Mismatch vertex} : if $x_{i+1} = [\leftP_{\ell+1} .. \rightP_{r+1}]$ and $S[\ell]  \neq T[r] $.
    \item \textbf{Left deletion vertex }: if $x_{i+1} = [\leftP_{\ell+1} .. \rightP_{r}]$.
    \item \textbf{Right deletion vertex} : if $x_{i+1} = [\leftP_{\ell} .. \rightP_{r+1}]$.
\end{enumerate}
Notice that since $p$ is well-behaved, $x_i$ is classified into one of the four types.

We proceed to analyze the cost of the subpath of $p$ from $x_i = x[\leftP_\ell .. \rightP_r]$ to $x_{i+1}$ using \cref{lem:constants}.
If $x_i$ is a match vertex, it induces a cost of $D_{\sync}(S[\ell])$.
If $x_i$ is a mismatch vertex, it induces a cost of $D(S[\ell]) + D(T[r])$.
If $x_i$ is a right (resp. left) deletion vertex, it induces a cost of $D(S[\ell])$ (resp. $D(T[r])$. 
For $\alpha,\beta \in \{ 0,1,2\}$ we present the following notations.
\begin{itemize}
    \item $c_{\match}(\alpha) = |\{ x\in \X \mid x \textit{ is a match vertex with } S[\ell]=\alpha\}|$
    \item $c_{\mis}(\{\alpha,\beta\})  = |\{ x\in \X \mid x \textit{ is a mismatch vertex with } \{S[\ell],T[r]\}=\{\alpha,\beta\} \}|$
    \item $c_{\mis}(\alpha)  = |\{ x\in \X \mid x \textit{ is a mismatch vertex with } S[\ell]=\alpha \textit{ or }T[r]=\alpha\}|$
    \item $c_{\LEFT}(\alpha)  = |\{ x\in \X \mid x \textit{ is a left deletion vertex with } S[\ell]=\alpha\}|$
    \item $c_{\RIGHT}(\alpha)  = |\{ x\in \X \mid x \textit{ is a right deletion vertex with } T[r]=\alpha\}|$
\end{itemize}
Note that since every super-gadget is deleted exactly once as a part of an $x_i$ to $x_{i+1}$ subpath.
It follows that for every $\alpha \in \{0,1,2\}$ we have $\#_\alpha(S) + \#_\alpha(T) = c_{\LEFT}(\alpha) + c_{\RIGHT}(\alpha) + c_{\mis}(\alpha) + 2c_{\match}(\alpha)$.
Note that for $\alpha\in\{0,1,2\}$  we have $c_\mis(\alpha)=\sum_{\beta\ne\alpha}c_\mis(\{\alpha,\beta\})$.
We denote $c_{\match} = c_\match(0) + c_\match(1)+c_\match(2)$.
We make the following claim:
\begin{claim*}
    $c_{\match} \le \LCS(S,T)$.
\end{claim*}
\begin{claimproof}
    We show that there is a common subsequence of $S$ and $T$ with length $c_{\match}$.
    Let $\Pairs= \{ ( \ell,r ) \mid x[\leftP_\ell .. \rightP_r] \textit{ is a match vertex} \}$.
    Note that $\Pairs$ is naturally ordered by the order of occurrences of the corresponding vertices in $p$.
    We denote by $(\ell_i,r_i)$ the $i$th pair in $\Pairs$ according to this order.
    Note that for every $i \in [|\Pairs| -1]$, we have $\ell_i < \ell_{i+1}$ and $r_i < r_{i+1}$ due to the definition of a match vertex.
    Furthermore, we have $S[\ell_i] = T[r_i]$ for every $i\in [|\Pairs|]$.
    It follows that the subsequence $S[\ell_1],S[\ell_2] \ldots ,S[\ell_{|\Pairs|}]$ equals to the subsequence $T[r_1],T[r_2] , \ldots ,T[r_{|\Pairs|}]$.
    Therefore, $S$ and $T$  have a common subsequence of length $|\Pairs|=c_\match$.
\end{claimproof}

It follows from the above analysis that 
\begin{align*}
\cost(p) &= 
\sum_{\alpha\in\{0,1,2\}} D_\sync(\alpha) c_{\match}(\alpha)+
\sum_{\alpha\ne\beta\in\{0,1,2\}}(D(\alpha)+D(\beta))\cdot c_{\mis}(\{\alpha,\beta\})
\\&\hspace{125px}+ \sum_{\alpha\in\{0,1,2\}}D(\alpha)(c_{\LEFT}(\alpha) + c_{\RIGHT}(\alpha))\\
&= 
\sum_{\alpha\in\{0,1,2\}} D_\sync(\alpha) c_{\match}(\alpha)+
\sum_{\alpha\in\{0,1,2\}}D(\alpha)\sum_{\beta\ne\alpha}c_{\mis}(\{\alpha,\beta\})
\\&\hspace{125px}+ \sum_{\alpha\in\{0,1,2\}}D(\alpha)(c_{\LEFT}(\alpha) + c_{\RIGHT}(\alpha))\\
&= 
\sum_{\alpha\in\{0,1,2\}} D_\sync(\alpha) c_{\match}(\alpha)+
\sum_{\alpha\in\{0,1,2\}}D(\alpha)\cdot c_{\mis}(\alpha)
\\&\hspace{125px}+ \sum_{\alpha\in\{0,1,2\}}D(\alpha)(c_{\LEFT}(\alpha) + c_{\RIGHT}(\alpha))\\&
= \sum_{\alpha\in\{0,1,2\}}(2D(\alpha) - B)c_{\match}(\alpha)  +
  \sum_{\alpha\in\{0,1,2\}}  D(\alpha)(c_\LEFT(\alpha) +c_\RIGHT(\alpha) + c_\mis(\alpha)) 
  \\&=
  \sum_{\alpha\in\{0,1,2\}}
  D(\alpha) (c_\LEFT(\alpha) +c_\RIGHT(\alpha) + c_\mis(\alpha) + 2c_\match(\alpha))
  -B\sum_{\alpha\in\{0,1,2\}}c_{\match}(\alpha) 
  \\ &=
    \sum_{\alpha\in\{0,1,2\}} D(\alpha) \cdot (\#_\alpha(S) + \#_\alpha(T)) - B \cdot c_\match 
  \\&\ge \sum_{\alpha\in\{0,1,2\}} D(\alpha) \cdot (\#_\alpha(S) + \#_\alpha(T)) - B \cdot \LCS(S,T).
\end{align*}
Where the last inequality follows from the claim.
\end{proof}

\bibliography{bib}

\appendix 
\section{Reduction to Hairpin Deletion}\label{app:equivalence} 
In this section, we close the gap between \cref{def:hairpin} and \cref{def:mod_hairpin} in order to prove \cref{thm:reduction}.
Let $\HDD'(x,y)$ be the hairpin deletion distance from $x$ to $y$ due to the original definition of hairpin operations (\cref{def:hairpin}).

We show a simple construction of strings $x'$ and $y'$ such that $\HDD(x,y) = \HDD'(x',y')$.
Thus, combining with \cref{lem:reduction}, finally proving \cref{thm:reduction}.

Let $\xmid\in [|x|]$ (resp. $\ymid \in [|y|]$) be the unique index such that $x[\xmid-1..\xmid+2] = 11\inv{11}$. (resp. $y[\ymid-1..\ymid+2] = 11\inv{11}$)
We define $x'$ of length $2|x|$ as follows.
For $i\le \xmid$ we set $x'[2i-1..2i]=2\cdot x[i]$ and for $i>\xmid$ we set $x'[2i-1..2i]=x[i]\cdot \inv{2}$.
Similarly, we define $y'$ of length $2|y|$ as follows.
For $i\le \ymid$ we set $y'[2i-1..2i]=2\cdot y[i]$ and for $i>\ymid$ we set $y'[2i-1..2i]=y[i]\cdot \inv{2}$.

For the purpose of showing that $\HDD(x,y)=\HDD'(x',y')$, we further define the \emph{Original Hairpin Deletion Graph $G'_{x'}=(V',E')$}, where $V'$ is the set of all substrings of $x'$ and there is an edge from $u$ to $v$ if $v$ can be obtained from $u$ by hairpin deletion operation, under \cref{def:hairpin}.

\begin{definition}[Original Hairpin Deletion Graph $G'_{x'}$]\label{def:original_graph}
    For a string $x'$ the \emph{Hairpin Deletion Graph} $G'_{x'}=(V,E)$ is defined as follows.
   $V$ is the set of all substrings of $x'$, and $(u,v)\in E$ if $v$ can be obtained from $u$ in a single hairpin deletion operation (as defined in \cref{def:hairpin}).  
\end{definition}
One can easily observe that $\HDD'(x',y')= \dist_{G'_{x'}}(x',y')$.
Therefore it remains to prove that $\dist_{G_{x}}(x,y)=\dist_{G'_{x'}}(x',y')$.

We present the following claims and observations, which have parallel versions in \cref{sec:well}. 

\begin{observation}[Parallel of \cref{obs:nodely}]\label{obs:nodelyprime}
    There is a unique occurrence of $y'$ in $x'$. 
    Let $y'=x'[i_y..j_y]$ be this unique occurrence.
    Let $p$ be a path from $x'$ to $y'$ in $G'_{x'}$. 
    For every vertex $x'[i..j]$ in $p$, $[i_y..j_y] \subseteq [i..j]$. 
\end{observation}

\begin{lemma}[Parallel of \cref{lem:01and10}]\label{lem:2and2}
Let $p$ be a path from $s$ to $t$ in $G'_{x'}$ such that $s,t \in 2 * \inv{2}$. 
For every vertex $x'[i..j]$ visited by $p$, we have $x[i]=2$ and $x[j] = \inv2$.
\end{lemma}
\cref{lem:2and2} can be proven similarly to \cref{lem:01and10}.

The following directly follows from \cref{lem:2and2}.
\begin{corollary}\label{cor:even_len}
Let $p'$ be a path from $x'$ to $y'$ in $G'_{x'}$, and let $x'[i..j]$ be a vertex visited by $p$.
Then, $i$ is odd and $j$ is even.
\end{corollary}

We are ready to prove the main lemma of this section.

\begin{lemma}\label{lem:equivlance}
    $\dist_{G_{x}}(x,y)=\dist_{G'_{x'}}(x',y')$
\end{lemma}

\begin{proof}
    We prove the lemma by showing $\dist_{G_{x}}(x,y)\ge\dist_{G'_{x'}}(x',y')$ and $\dist_{G_{x}}(x,y)\le\dist_{G'_{x'}}(x',y')$.

    \para{\underline{$\dist_{G_{x}}(x,y)\ge\dist_{G'_{x'}}(x',y')$}}
    Let $p=(x=u_1,\ldots ,u_k=y)$ be a shortest path from $x$ to $y$ in $G_{x}$.
    For $a\in[k]$ let $u_a=x[i_a..j_a]$ and we define $u'_a=x'[2i_a-1..2j_a]$.
    We next show that for $a\in[k-1]$ the edge $(u'_a,u'_{a+1})$ is in $G'_{x'}$ and therefore $p'=(u'_1,u'_2,\ldots,u'_k)$ is an $x'$ to $y'$ path of length $k=\dist_{G_x}(x,y)$ in $G'_{x'}$.

    Assume without loss of generality that $(u_a,u_{a+1})$ is a left edge in $G_x$. 
    Therefore, we have $x[i_a..i_a + \ell - 1] = \rev{x[j_a -\ell+1 .. j_a]}$ for $\ell = i_{a+1} - i_a$.
    By our construction, $x'[2i_a-1..2(i_a + \ell - 1)] = \rev{x'[2(j_a -\ell+1)-1 .. 2j_a]}$ (the reduction simply added $2$ symbols in the odd positions of the left substring and $\inv{2}$ in the even positions of the right substring).
    Since $x'[2(i_a + \ell - 1)+1]=2$ and $x'[2(j_a -\ell+1)-2]=\inv{2}$ it follows from \cref{def:hairpin} that there exists a valid left hairpin deletion operation of length $2\ell$ as required.
    
    \para{\underline{$\dist_{G_{x}}(x,y)\le\dist_{G'_{x'}}(x',y')$}}
    Let $p'=(x'=u'_1,u'_2,\ldots,u'_k=y')$ be a shortest path from $x'$ to $y'$ in $G'_{x'}$.
    By \cref{cor:even_len} for every $a\in[k]$ we have $u'_a=x'[2i_a-1..2j_a]$ for some integers $i_a,j_a$.
    For every $a\in [k]$ we define $u_a = x[i_a..j_a]$.

    We prove that for every $a \in [k-1]$, the edge $(u_a,u_{a+1})$ exists in $G_x$, and therefore the path $p=(x=u_1,u_2, \ldots ,u_k)$ is an $x$ to $y$ path in $G_x$ with length $k = \dist_{G'_{x'}}(x',y')$.

   Assume without loss of generality that $(u'_a,u'_{a+1})$ is a left edge (corresponding to a left hairpin deletion) in $G'_{x'}$. 
    Therefore, we have $x'[2i_a-1..2i_a - 1 + \ell' - 1] = \rev{x'[2j_a -\ell'+1 .. 2j_a]}$ for $\ell' = 2i_{a+1} - 2i_a$.
    Note that $\ell' = 2\ell$ for some integer $\ell$.
    By our construction, $x[i_a..i_a + \ell - 1] = \rev{x[j_a -\ell+1 .. j_a]}$.
    It follows from \cref{def:mod_hairpin} that there exists a valid left hairpin deletion operation of length $\ell$ as required.
    
\end{proof}

\section{Missing Proofs from \texorpdfstring{\cref{sec:well}}{Section 4}}\label{app:missingwell}

\subsection{Missing proofs of \texorpdfstring{\cref{sec:local}}{Section 4.1}}\label{app:missinglocl}
\nodelyObs*
\begin{proof}
    First, notice that the occurrence of $y$ is indeed unique since $y$ contains the only occurrence of $11\inv{11}$ in $x$ (to the left of $y$ there is no occurrence of $11$ and to the right of $y$ there is no occurrence of $\inv{11}$).
    It immediately follows that if a sequence of hairpin deletion operations deletes a segment of $y$ (which corresponds to visiting a vertex $x[i..j]$ with $[\leftP_{|S|+1} .. \rightP_{|T|+1}] \not\subseteq [i..j]$), it can no longer reach $y$. 
\end{proof}

\lemzerooneandonezero*
\begin{proof}
    We start by proving the first statement.
    Assume by contradiction that the lemma is false and consider the first vertex  $u = x[i..j]$ visited by $p$ with $x[i]=1$ or $x[j]=\inv1$, without loss of generality, assume that $x[i]=1$.
    By the minimality of $u$, we have $x[j]=\inv0 \ne \inv{x[i]}$. 
    It follows that no hairpin deletion operation can be applied on $u$ and there are no outgoing edges in $G_x$ from $u$.
    Note that $u\neq t$, which contradicts the assumption that $p$ is a path to $y$.
    
    We proceed to prove the second statement.
    By contradiction, consider the first vertex $v = x[i..j]$ visited by $p$  with $x[i+1] =0$ and $x[j-1] = \inv0$.
    Assume without loss of generality that the hairpin deletion leading to $v$ is a left hairpin operation, so we have that the vertex precedes $v$ in $p$ is $u = x[i'..j]$ for some $i' < i$ that satisfies $x[i'+1] = 1$ (note that in the initial vertex $s$ satisfies the condition of the lemma). 
    Recall that $x[j-1..j] = \inv{00}$ and $x[i'..i'+1] =01$.
    Therefore, the only valid left hairpin deletion operation on $v$ is of length $1$, resulting in $i = i'+1$ and $x[i] = 1$, contradicting the first claim.
\end{proof}

\lemretzerooneandonezero*
\begin{proof}
    We prove the first statement; the second statement can be proved symmetrically.
    If $k=1$, we are done. 
    Notice that for a string of the form $01 * \inv{10^k}$ with $k \ge 2$, the only legal hairpin deletion operations are of length $1$ (either from the left or from the right).
    A deletion from the left of length one would result in $x[i+1..j] \in 1*$, which contradicts \cref{lem:01and10}.
    Therefore, $p$ visits $x[i..j-1]$.
    By induction on $k'$, the path $p$ eventually visits $x[i..j-k']$.
\end{proof}

\lemsync*
\begin{proof}
We prove the first statement; the second statement can be proven similarly.
Note that every $\SL$ outside of $y$ is followed by a left gadget starting with $01$, therefore $x[i_s.. j_s + 2] = \SL \cdot 01 =0101$.
Let $u= x[i..j]$ be the first vertex in $p$ with $i > i_s + 1$.
Notice that $u$ is preceded in $p$ by $v=x[i'..j]$ for some $i' < i$  (in particular, $u$ is obtainable from $v$ via a single left hairpin deletion operation).
We claim that $i \le  j_s + 2$. 
Assume to the contrary that $i \ge j_s + 3$.
Due to $u$ being obtainable from $v$, we have $x[i'..j_s+2] = \rev{x[j'..j]}$ for some index $j'\in [|x|]$.
In particular, a single hairpin deletion operation was applied by $p$ to delete a superstring of $x[i_s+1..j_s+2] = 101$.
Due to \cref{obs:nodely}, $j$ is to the right of $y$.
By the construction of $x$, there is no occurrence of $\rev{101} = \overline{101}$ to the right of the central $11\inv{11}$ in $y$.
It follows that $j'$ is an index to the left of the occurrence of $11\overline{11}$ in $y$.
In particular, $x[j'..j]$ contains the substring $11\overline{11}$.
Recall that $x[i'..j_s+2]$ is to the left of $y$ in $x$, and that there is no occurrence of $\rev{11\overline{11}} =11\overline{11}$ in $x$.
This contradicts the equality $x[i'..j_s+2] = \rev{x[j'..j]}$.

We have shown that $i \le j_s+2$.
Due to \cref{lem:01and10} and $x[j_s+2] = 1$ we have $i\ne j_s+2$, thus $i<j_s+2$.
It follows that $x[i..j_s+2] = 0^z1$ for $z=j_s + 2 -i\ge 1$.
It follows from \cref{lem:ret0110} that $p$ visits $x[j_s +1 ..j]$ as required.
\end{proof}

\subsection{Missing proofs of \texorpdfstring{\cref{sec:trans}}{Section 4.2}}\label{app:missingtrans}
\lemidnotseer*
\begin{proof}
We prove the first statement, the proof for the second statement follows symmetrically.
Assume to the contrary that there is $r\in[|T|]$ such that the $r$th  $\PR$ gadget is contained in $x[j_2 .. j_1]$

Due to \cref{cor:sync}, the path $p$ visits a vertex $z = x[i' .. \rightP_{r}]$ for some $i'$.
In particular, since $\rightP_{r} \in [j_2 .. j_1]$ we must have $i' \in [\leftI_\ell .. \leftP_{\ell+1}]$.
Due to \cref{lem:ret0110}, we can assume that $x[i'..i'+1] = 01$.
Notice that $z$ occurs between $v$ and $u$ in $p$.
Consider the sub-path $q$ between $z$ and $u$ in $p$. 
We show an alternative path $q^*$ from $z$ to $u$ that is strictly shorter than $q$, which contradicts the minimality of $p$.

Let $x[a..b]$ be an internal vertex in $q$ such that $b$ is an index in some $\PR$ gadget and the following vertex in $q$ is $x[a..b']$ with $b' < b$.
Note that every symbol deletion from a $\PR$ gadget in $q$ corresponds to such vertex.  
Since $a \in [\leftI_\ell .. \leftP_{\ell+1} - 1]$, \cref{clm:single1char} suggests that at most a single $\inv1$ character of the $\PR$ gadget is deleted in this operation.
It follows that every $\PR$ gadget contained in $[j_2.. j_1]$, induces a cost of at least $\#_{\inv1}(\PR) = \p$ for $q$. 
Let $k_\prot$ be the number of such $\PR$ gadgets.
The above reasoning shows that $\cost(q) \ge k_\prot \cdot \p$.

We proceed to present $q^*$.
Initially, $q^*$ deletes the prefix $x[i' .. \leftP_{\ell+1}-1]$ consisting of a suffix of $\IL$ and $\SL$, a single $'1'$ at a time.
Notice that this is possible (regardless of $\ell$th left information gadget being $\IL(0)$, $\IL(1)$ or $\IL(2)$) due to $x[\rightP_{r} - 8 .. \rightP_{r}] = \inv{0^710}$.
This induces a cost of at least $\#_{1}(\IL(\alpha)) + 1 = \i_\alpha + 1 \le \i_0 + 1$ (for $\alpha=S[\ell]$).
The rest of the path from $z'=x[\leftP_{\ell+1} .. \rightP_{r}]$ to $u=x[\leftP_{\ell+1} .. j_2]$  removes $x[j_2+1..\rightP_{r}]$, we denote $I_\remove=[j_2+1..\rightP_{r}]$.
We describe the rest of $q^*$ from $z'$ as a concatenation of sub-paths from $x[\leftP_{\ell+1} .. \rightP_{a}]$ to $x[\leftP_{\ell+1} .. \max\{\rightP_{a+1},j_2\}]$ such that $\rightP_{a}\in I_\remove$.

If the $a$th $\PR$ gadget and the following $\SR$ gadgets are fully contained in $I_\remove$, then $q^*$ deletes the $a$th $\PR$ gadget and the following $\SR$ gadget in one step, reaching $\rightI_a$.
This is possible since $x[\leftP_{\ell+1}..\leftI_{\ell+1}]=\rev{\SR\cdot\PR}$.
Otherwise, $q^*$ removes $x[j_2+1..\rightP_{a}]$ in one step.
$q^*$ proceeds to deal with the $a$th right information gadget and the following $\SR$ gadget.
If the $a$th right information gadget and the following $\SR$ gadgets are fully contained in $I_\remove$, then $q^*$ deletes the $a$th $\IR$ gadget, a single $\inv 1$ at a time until reaching the next $\SR$ gadget, $q^*$ then deletes the $\SR$ gadget using an additional hairpin deletion operation.
Notice that this is possible 
(regardless of $a$th right information gadget being $\IR(0)$, $\IR(1)$ or $\IR(2)$) due to $x[\leftP_{\ell+1} ..  \leftP_{\ell+1}  + 8] = 010^7$.
These deletions induce a cost of $\#_{\inv1}(\IR(T[a])) + 1$ to $p^*$.
Otherwise, if the $a$th right information gadget or the following $\SR$ gadget contains $j_2$, then $q^*$ removes $x[j_2+1..\rightI_{a}]$ in the same fashion, in at most $\#_{\inv1}(\IR(T[a])) + 1 \le \i_0 + 1$ steps.
This concludes the construction of $q^*$.

We proceed to bound the cost of $q^*$. 
Denote the number of $\IR$ gadgets with indices in $I_{\remove}$ as $k_\info$.
Note that since $x[j_2+1 .. \rightP_{r}]$ (the substring of $x$ corresponding to $I_\remove$) ends with a $\PR$ gadget, we have $k_\info \le k_\prot$.
It follows from the above analysis that the cost of $q^*$ is at most 
\[\cost(q^*)\le  ( k_\prot+1) + (\i_0 + 1) \cdot k_\info + \i_0 +1
\le (\i_0 + 2)\cdot k_\prot + \i_0 +2 .\]
Recall that $\cost(q) \ge  k_\prot \cdot \p$. 
Since $144=\p > 2\cdot \i_0 + 5=115$ and $k_\prot \ge 1$, we have that $\cost(q) > \cost(q^*)$, this concludes the proof.
\end{proof}

\lemprotDelay*
\begin{proof}
We prove the first statement, and the rest of the statements can be proven similarly.
Since $(v,u)$ is an edge in $G_x$ we have $x[i..i+k - 1]=\rev{x[j-k+1..j]}$.
We distinguish between two cases.
First, consider the case when $x[i..i+1]\ne 01$.
From \cref{lem:01and10}  and our assumption that $p$ visits $x[i.. j]$, we have that $x[i..i+1] = 00$ and $x[j-1..j] = \inv{10}$.
Therefore, we must have $k = 1$.
Since $x[\leftP_\ell]= 0$ and $x[j] = \inv0$, the edge $(v',u')$ is in $G_x$.

Next, consider the case where $x[i..i+1]=01$. For this case, we make the following claim regarding $k$.
\begin{claim*}
$k\le \leftI_\ell-i$.
\end{claim*}
\begin{claimproof}
    Assume by contradiction that $k>\leftI_\ell-i$. 
    Since $x[i..i+1]=01$, it must be that $i\le \leftI_\ell-|\SL|$ (the location of the last $01$ in $x[\leftP_\ell..\leftI_\ell-1])$. 
    Thus, $x[i..i+k]$ contains $x[\leftI_\ell-|\SL|+1..\leftI_\ell+1]=101$ (for the cases regarding $\PR$ and $\IR(\alpha)$ we have $x[j-k..j]$ contains $\rev{1001}$).
    Notice that $\rev{101}$ does not appear to the right of $y$.
    Contradicts the assumption  $x[i..i+k - 1]=\rev{x[j-k + 1..j]}$.
\end{claimproof}

Thus, it remains to prove that $x[\leftP_\ell..\leftP_\ell+k-1]=x[i..i+k-1]$.
Recall that we are considering the case in which $x[i..i+1] = 01$. 
Recall that $\PL = (010^9)^{\p}$ is periodic with period of length $11$ (for the case regarding $\IL(\alpha)$ we have $\IL(0)=010^3$, $\IL(1)=010^5$ and $\IL(2)=010^7$; The cases regarding $\PR,\IR(\alpha)$ are symmetric).
Thus, in this case, we must have $i = \leftP_\ell + q\cdot 11$ for some non-negative integer $q$.
Furthermore, since $x[\leftI_\ell-|\SL|..\leftI_\ell]=\SL \cdot 0 = 010$ is a prefix of the period of $\PL$, we have that $x[\leftP_\ell .. \leftI_\ell]$ is periodic with period of length $11$.
Due to $i+ k \le \leftI_\ell$ (according to the claim), we have $x[\leftP_\ell .. \leftP_\ell + k-1] = x[i .. i+k -1]$.
It follows that $x[\leftP_\ell .. \leftP_\ell + k - 1] = \rev{x[j-k + 1 .. j]}$ and $(v',u')$ is an edge an $G_x$, as required.
\end{proof}

\section{Analysis of Deletion of Synchronized Gadgets}\label{app:FibDel}
This section is dedicated to proving \cref{lem:FibDistance}.
We start by providing some intuition.
As an exercise, consider the hairpin distance from $01\inv 0$ to $0^a1\inv{0}^a$.
It turns out that this distance is closely related to the Fibonacci sequence.
It is easier to observe this relationship when considering the problem in terms of hairpin completion distance.
In this setting, a Fibonacci sequence represented a 'greedy' sequence of completions - taking the larger number of $0$ or $\inv{0}$ symbols and using them to 'complete' the other side.
The following lemma is more technical, fitting our reduction, that in a sense claims that after $x$ hairpin completion operations starting with $01\inv 0$ and ending with $0^a1\inv0^b$ we have $\min\{a,b\}\le\Fib(x)$ and $\max\{a,b\}\le\Fib(x+1)$.

\begin{lemma}\label{lem:fib}
Let $\per=\boxed{010^\ext}$ for $\ext\ge 3$ and $q\in\boxed{0101*11\inv{11}*\inv{10010}}$ be two strings.
Let $k$ be a positive integer and let $a,b \ge k$ be two positive integers. 
Starting with $\per^k \cdot q\cdot \rev{\per}^k$,  after $x$ hairpin completion operations ending with $\per^a\cdot q\cdot \rev{\per}^b$, we have $\min\{a,b\}\le k \cdot \Fib(x)$ and $\max\{a,b\}\le k \cdot \Fib(x+1)$.
\end{lemma}
\begin{proof}
For a sequence of $x$ hairpin completion operations, let $S_i$ be the string after applying the first $i$ operations.
First, observe that $q$ appears only once in $\per^a\cdot q\cdot \rev{\per}^b$, since $q$ contains the string $11\inv{11}$, and $\per^a,\rev{\per}^b$ do not.
Thus, at every $S_i$, $q$ appears exactly once and $S_i$ has a unique parsing $S_i = X\cdot q\cdot Y$ for some strings $X$ and $Y$.

We prove the claim by induction on $x$.
As a base, for $x=0$ the claim follows.
Assume the claim holds for all natural numbers up to $x$ and we prove for $x+1$. 
Consider a sequence of $x+1$ hairpin completion operations starting with $\per^k \cdot q \cdot \rev{\per}^k$, ending with $\per^a \cdot q \cdot \rev{\per}^b$.
Without loss of generality, we assume that the $(x+1)$th operation was a right hairpin completion.
We consider two cases regarding $S_x$.
\para{Case 1 $S_x=\per^a \cdot q \cdot \rev{\per}^{b'}$ for some $b'<b$:}
Notice that the $(x+1)$th right hairpin completion copies exactly $|\per|\cdot(b-b')$ characters. 
We claim that $b-b' \le a$.
Assume to the contrary that $b-b'>a$, in particular, the last operation copies the prefix of length $|\per|$ of $q$ (since $(b-b')\in\N$).
This prefix contains the string $101$ such that $\rev{101} = \inv{101}$ is not included in $\rev{\per}^b$ (since $\inv{\per}=\inv{010^\ext}$ and $\ext\ge 3$) contrary to our case (in the symmetric case of left hairpin completion operation, a suffix of $q$ is copied which contains $\rev{1001}$ as a substring).
To conclude, we just proved that $b\le a+b'$.
By the induction assumption,  $\min\{a,b'\}\le k\cdot \Fib(x)$ and $\max\{a,b'\}\le k\cdot  \Fib(x+1)$.
Thus, $\min \{a,b\}\le\min\{a,a+b'\}= a \le \max\{a,b'\}\le k\cdot \Fib(x+1)$, and $\max\{a,b\}\le\max\{a,a+b'\}=a+b'\le k\cdot (\Fib(x)+\Fib(x+1))= k\cdot \Fib(x+2)$.

\para{Case 2 $S_x=\per^a \cdot q \cdot \rev{\per}^{b'}\cdot \per'$ where $b'<b$ and $\per'$ is a proper prefix of $\rev \per$:} 

Let $y$ be the maximal index in $[x]$ such that $S_y \in * 10$.
We claim that $S_y = \per^a \cdot q \cdot \rev{\per}^{b'}$ and $y = x-|\per'|$.
We prove the claim by showing that for every $z\in [y+1..x]$, the $z$th operation is a right hairpin completion of length $1$.
Note that for $z\in [y+1 .. x]$ we have $S_z = * \inv{00}$, and therefore it must also satisfy $S_z = 01 * \inv{00}$ according to \cref{cor:0110deletion}.
Clearly, the $z$th operation is not a left hairpin completion operation, as this would result in $S_{z} \in 00 * \inv{00}$ which contradicts \cref{cor:0110deletion}.

Since $\inv{S_z[2]} \neq S_z[|S_z| - 1]$, the only valid right hairpin completion on $S_{z-1}$ is of length $1$ as required.
Therefore, we have $S_{x-|\per'|} = \per^a \cdot q \cdot \rev{\per}^{b'}$ and by the induction assumption we have $\min(a,b') \le k \cdot \Fib(x-|\per'|)$ and $\max(c,d') \le k \cdot \Fib(x-|\per'|+1)$.

Notice that the $(x+1)$th operation copies $|\per|(b-b') - |\per'|$ characters. 
We claim that $b-b' \le a+1$.
Assume to the contrary that $b-b' > a+1$.
Recall that $|\per'| \le |\per|-1$.
It follows that, the $(x+1)$th operation copies a prefix of length $|\per|$ of $q$, which is a contradiction due to the same reasoning as in case 1.

To conclude, we just proved that $b\le a+b' + 1$.
By induction assumption,  $\min\{a,b'\}\le k\cdot \Fib(x-|\per'|)$ and $\max\{a,b'\}\le k \cdot \Fib(x- |\per'| + 1)$.
Thus, $\min \{a,b\}\le\min\{a,a+b' + 1\}= a \le \max\{a,b'\}\le k\cdot \Fib(x-|\per'|+1)\le k \cdot \Fib(x+1)$, and $\max\{a,b\}\le\max\{a,a+b'+1\}=a+b'+1\le k \cdot (\Fib(x-|\per'|)+ \Fib(x-|\per'|+1)) + 1 = k\cdot \Fib(x-|\per'|+2) + 1 \overset{(*)}{\le} k\cdot\Fib(x+1) + 1 \le k\cdot \Fib(x+2)$ where $(*)$ is due to $|\per'| \ge 1$.
This concludes the proof.
\end{proof}

The following lemma uses \cref{lem:fib} to compute the hairpin completion distance from $\per^k\cdot  q \cdot \rev{\per}^k$ to $\per^a\cdot  q \cdot \rev{\per}^a$ for every two integers $a\ge k$.

\begin{lemma}\label{lem:hairpinfib}
Let $\per=010^\ext$ with $\ext\ge 3$ and let $q\in0101*11\inv{11}*\inv{10010}$.
For every two positive integers $a \ge k$, we have  
$\HC(\per^k\cdot  q \cdot \rev{\per}^k,\per^a\cdot  q \cdot \rev{\per}^a)=\Fib^{-1}(\frac{a}{k})$.
\end{lemma}
\begin{proof}
Assume that a sequence of $x$ hairpin completion operations transforms $\per^k\cdot q \cdot \rev{\per}^k$ into $\per^a \cdot q \cdot \rev{\per}^a$.
For $i\in [x]$, let $S_i$ be the string after applying the first $i$ operations.
Assume without loss of generality that the last operation in the sequence is a right hairpin completion.
Let $t$ be the minimal index in $[x]$ with $S_t = \per^a \cdot q \cdot Y$ for some prefix $Y$ of $\rev{\per}^a$.

\begin{claim*}
$Y=\rev{p}^b$ for some $b \le a$.
\end{claim*}
\begin{claimproof}
    
Clearly, the $t$th operation is a left hairpin completion due to the minimality of $t$.
Since $S_t[1..2] = 01$, and $S_{t-1}[1] \neq 1$ (due to \cref{cor:0110deletion}) it must be the case that a prefix starting with $01$ was prepended to $S_t$ in the $t$th operation.
Therefore, $S_t$ ends with $\inv{10}$.
The only prefixes of $\rev{\per}^a$ that end with $\inv{10}$ are of length $b \cdot |\per|$ for some integer $b$, the claim follows directly.
\end{claimproof}

We have shown that $S_t = \per^a \cdot q \cdot \rev{\per}^b$ for some integer $b \le a$. 
By \cref{lem:fib}, $a \le k \cdot \Fib(t+1)$.
Therefore, $t \ge \Fib^{-1}(\frac{a}{k}) - 1$.
Since the $t$th operation is a left hairpin completion, and the sequence ends with a right hairpin completion we have $x \ge t+1 =  \Fib^{-1}(\frac{a}{k})$ and therefore $\HC(\per^k\cdot  q \cdot \rev{\per}^k,\per^a\cdot  q \cdot \rev{\per}^a) \ge\Fib^{-1}(\frac{a}{k})$.

We proceed to show a sequence of $\Fib^{-1}(\frac{a}{k})$ hairpin completion operations that transform $\per^k\cdot  q \cdot \rev{\per}^k$ into $\per^a\cdot  q \cdot \rev{\per}^a$.

Our sequence is simply a greedy sequence.
For $S_i = \per^c \cdot q \cdot \rev{\per}^d$ such that $c \ge d$, we proceed the sequence with $S_{i+1} = \per^{c} \cdot q \cdot \rev{\per}^{\min(c+d,a)}$.
Symmetrically, if $c < d$ we proceed with $S_{i+1} = \per^{\min(c+d,a)} \cdot q \cdot \rev{\per}^d$.
The sequence terminates upon reaching $S_t = \per^a \cdot q \cdot \rev{\per}^a$.
Note that our sequence satisfies $S_i = \per^c \cdot q \cdot \rev{\per}^d$ such that $\min(c,d) = \min( k \cdot \Fib(i),a)$ and $\max(c,d) = 
\min( k\cdot \Fib(i+1), a)$.
It follows that our sequence terminates when $k\cdot \Fib(i) \ge a$.
By definition, this is satisfied when $i = \Fib^{-1}(\frac{a}{k})$.
\end{proof}

Let  $\per=010^{\ext}$ and $q\in010^{\int}01*11\inv{11}*\inv{100^{\int}10}$ be two substrings of $x$ for some $\int \neq \ext\in\N$.
Let $p$ be a shortest path from $v = \per^a\cdot\SL\cdot  q\cdot \SR \cdot \rev{\per}^a$ to $u = q$ in $G_x$.
We denote by $\last_{\per,q}(p)$ the last vertex in $p$ with $v_1 \in * \per\cdot\SL\cdot  q\cdot \SR \cdot \rev{\per} *$.
In words, $\last_{\per,q}(p)$ is the last vertex in $p$ that contains at least a complete occurrence of $\per$ to the left of $q$ and a complete occurrence of $\rev\per$ to the right of $q$.
When $\per$ and $q$ are clear from the context, we simply write $\last(p)$.

We prove the following lemma regarding the suffix of a path $p$ from $\last(p)$ to $u=q$.
\begin{lemma}\label{clm:shortdeletions}
    Let $\per=010^\ext$ and let $q\in010^\int01*11\inv{11}*\inv{100^\int10}$ with $\int\ne \ext \ge 3$.
    Let $p$ be a shortest path from $v$ to $u$ and let $v_1=\last(p)$.
    If the operation applied by $p$ on $v$ is a right (resp. left) hairpin deletion then every left (resp. right) hairpin deletion in the subpath from $v_1$ to $u$ is of length at most $|\per| = \ext+2$. 
\end{lemma}
\begin{proof}
    We prove the case where $p$ applies a right hairpin deletion to $v_1$, the proof of the other case is symmetric.
    Consider a vertex $v'=x[i'.. j']$ in $p$ such that the following vertex in $p$ is $w' = [i'+d .. j']$.
    If $d=1$, we are done.
    Otherwise,
    \cref{obs:stableside01} suggests that  $x[i' .. i'+1 ] = 01$ and $x[j'-1..j']= \inv{10}$.
    Assume to the contrary that $d > \ext +2$.
    Note that every substring of $\per^a \cdot \SL \cdot q \cdot \SR \cdot \rev{\per}^a$ that is in $01*q*$, is also in $\per^z \cdot \SL \cdot q *$ for some $z \in [0..a]$.
    Note that if $z = 0$, a left deletion operation of length at least $\ext +2$ is impossible since $|\SL| = 2 < \ext +2$, and therefore the left deletion would delete a prefix of $q$.
    Therefore, $z \ge 1$.
    Note that it can not be the case that $d=\ext +3$, since that would result in $w'[1] = 1$, contradicting \cref{lem:01and10}.
    Therefore we have $d \ge \ext+4$.
    Since $v'$ is strictly after $v_1$ in $p$, we have $v' \notin  * q \cdot \SL \cdot \rev{\per}$.
    It follows that $v'$ satisfies either $v' \in *q$ or $v' \in * q \cdot \SR$.
    In both cases one of the following holds:
    \begin{enumerate}
        \item $v'[j' - \ext - 3 ..j']$ contains a substring $\inv{10^{\int+1} 1}$ (if $\int < \ext$ and $v'$ ends with $q$).
        \item $v'[j - \ext - 3 ..j']$ contains a substring $\inv{0^{\ext+2}}$ (if $\int > \ext $ and $v'$ ends with $q$).
        \item $v'[j' - \ext - 3 ..j']$ contains a substring $\inv{1001}$
        (if $v'$ ends with $\SR$)
    \end{enumerate}
     In all cases, $v'[j' - \ext - 3 ..j']$ contains a substring such that the reverse (of the inverse) of this substring does not occur to the left of $q$ in $v$.
     This contradicts the assumption that the hairpin deletion from $x[i'..j']$ to $x[i' +d .. j']$ is valid.
\end{proof}

The following lemma allows us to assume that $\last(p)$ has a particular structure.

\begin{lemma}\label{lem:lastkk}
Let $\per=010^\ext$ and let $q\in010^\int01*11\inv{11}*\inv{100^\int10}$ with $\int\ne \ext \ge 3$.
There exists a  shortest path $p^*$ in $G_x$ from $\per^a\cdot\SL\cdot  q\cdot \SR \cdot \rev{\per}^a$ to $q$ such that  $\last(p^*)=\per^k\cdot\SL\cdot  q\cdot \SR \cdot \rev{\per}^k$ for some $k\in[a]$.
\end{lemma}
\begin{proof}
    Let $p$ be a shortest path from $v= \per^a\cdot\SL\cdot  q\cdot \SR \cdot \rev{\per}^a$ to $u = q$ in $G_x$.
    We assume that the operation that follows $v_1 = \last(p)$ in $p$ is a right hairpin deletion.
    The case in which the operation following $v_1$ is a left hairpin deletion is symmetrical.
    We present a path $p^*$ from $v$ to $u$ that is not longer than $p$ such that $\last(p^*)= \per^k\cdot\SL\cdot  q\cdot \SR \cdot \rev{\per}^k$ for some positive integer $k \in [a]$.
    As a first step, we prove the following claim.

    \begin{claim}
    $v_1=\per^{k_1}\cdot\SL\cdot  q\cdot \SR \cdot \rev{\per}^{k_2}$ for some $k_1,k_2\in\N$.
    \end{claim}
    \begin{claimproof}
        We start by showing that the right hairpin deletion applied to $v_1$ in $p$ is not of length $1$.
        If $v_1[|v_1| - 1] = \inv{1}$, a right hairpin deletion with length $1$ would result in a string ending in $\inv{1}$, contradicting \cref{lem:01and10}.
        Otherwise, by \cref{lem:01and10} we have $v_1[|v_1| - 1..|v_1| ] = \inv{00}$. 
        Therefore, a right deletion of a single $\inv{0}$ could not result in a string not in $* \per\cdot\SL\cdot  q\cdot \SR \cdot \rev{\per} *$,  contradicting the maximality of $v_1$ as $\last(p)$.
        Therefore, the right hairpin deletion applied to $v_1$ has length at least $2$. 
        
        Since a hairpin deletion operation of length at least $2$ can be applied to $v_1
        $, \cref{lem:ret0110}, implies that $v_1 \in 01 * \inv{10}$.
        Recall that  $v_1$ is a substring of $ \per^a\cdot\SL\cdot  q\cdot \SR \cdot \rev{\per}^a$ and $v_1 \in * \per\cdot\SL\cdot  q\cdot \SR \cdot \rev{\per}*$ with $\per=010^\ext $.
        Every string that satisfies both constraints and is in $01 * \inv{10}$ is of the form $\per^{k_1} \cdot \SL \cdot q \cdot \SR \cdot \rev{\per}^{k_2}$ for some integers $k_1$ and $k_2$, as required.
    \end{claimproof}

    If $k_1=k_2$ we are finished by setting $p^* = p$.
    Otherwise, we distinguish between two cases:

\para{Case 1 $k_1>k_2$:}
By \cref{clm:shortdeletions,lem:ret0110}, for every $k<k_1$ the path $p$ must visit a vertex $w_k\in\per^k\cdot\SL\cdot q *$ after visiting $v_1$.
In particular $p$ visits $w_{k_2}=\per^{k_2}\cdot\SL\cdot q *$.
The path $p^*$ is composed of three parts. 
The prefix of $p^*$ is identical to the prefix of $p$ from $v$ to $v_1$.
The suffix of $p^*$ is identical to the suffix of $p$ from $w_{k_2}$ to $u$.
Now, we describe the subpath $p^*_{\middle}$  of $p^*$ from $v_1$ to $w_{k_2}$ that replaces the subpath $p_{\middle}$ of $p$.
The path $p^*_{\middle}$ begins with $k_1-k_2$ left hairpin deletion operations, each one of them is of length $|\per|=\ext+2$.
We point out that at this point $p^*$ visits the node $v'=\per^{k_2} \cdot \SL \cdot q \cdot \SR \cdot \rev{\per}^{k_2}$.
Then, $p^*$ move from $v'$ to $w_{k_2}$ in a single step.
Notice that such a right hairpin deletion operation is allowed since $\rev{\per^{k_2}\cdot\SL\cdot 0}=\SR\cdot\rev{\per}^{k_2}$.
Note that $\cost(p^*_{\middle})\le\cost(p_{\middle})$ since by \cref{clm:shortdeletions}, the number of left hairpin deletions in $p_{\middle}$ is at least $k_1-k_2$, and there is at least one right hairpin deletion in $p_{\middle}$ (following $v_1$).
Thus, $\cost(p_{\middle})\ge k_1-k_2+1=\cost(p^*_{\middle})$.
Finally, $\cost(p^*)\le\cost(p)$ and $\last(p^*) = \per^{k_2} \cdot \SL \cdot q \cdot \SR \cdot \rev{\per}^{k_2}$, as required.

\para{Case 2 $k_1< k_2$:}
Denote by $\len$ the length of the right hairpin deletion applied to $v_1$ in $p$.
Since $v_1 = \last(p)$, we have $\len > (k_2-1)\cdot |\per| \ge k_1 \cdot |\per| $.
Since $k_1 < k_2$,  the string $v_1$ has a prefix $\per^{k_1} \cdot 0101$ and a suffix $\inv{0010} \cdot \rev{\per}^{k_1}$.
Therefore, $\len \le k_1 \cdot |\per| + 3$.
Let $v_2$ be the vertex following $v_1$ in $p$.
It follows from the above discussion that $v_2 =  \per^{k_1} \cdot \SL \cdot q \cdot \SR  \cdot \rev{\per}'$ for some proper prefix $\rev{\per}'$ of length at least $|\per| - 3$ of $\rev{\per}$. 
Due to \cref{lem:01and10}, $v_2$ can not end with $1$, so $|\rev{\per}'| \le |\per| - 2$.
In particular, $v_2$ has a suffix $\inv{0^{\ext}}$ , more specifically
\footnote{In the symmetric case where the operation following $\last(p)$ is a \emph{left} hairpin deletion, in the this case ($k_1>k_2$) we get that $\last(p)$ has a prefix $\per^{k_2}\cdot01000$ and a suffix $\inv{10010}\cdot\rev{\per}^{k_2}$. Therefore $\len\le k_2\cdot|\per|+4$. $\per'$ is a suffix of $\per$ of length at least $|\per|-4$. Thus $v_2$ has a prefix $0^{\ext-1}$ and in particular a prefix $0^2$.}
, $v_2$ has a suffix $\inv{0^3}$.
Due to \cref{lem:ret0110}, we have that the vertex following $v_2$ in $p$ is $v_3 =  \per^{k_1} \cdot \SL \cdot q \cdot \SR  \cdot \rev{\per}'[1..|\rev{\per'}| -1]$.

We are now ready to describe $p^*$.
The path $p^*$ is almost identical to $p$, except replacing the subpath of length $2$ from $v_1$ to $v_3$ with the following subpath:
$v_1$ is followed by $v'_2 = \per^{k_1} \cdot \SL \cdot q \cdot \SR \cdot \rev{\per}^{k_1}$.
The vertex $v'_2$ is then followed by $v_3$.
Note that there is an edge from $v_1$ to $v'_2$ since it requires a right hairpin deletion operation with length at most $\len$, and a right hairpin deletion from $v_1$ of length $\len$ is possible. 
Additionally, we claim that there is an edge from $v'_2$ to $v_3$ in $G_x$.
This is true since $v'_2$ has a prefix $\per^{k_1} \cdot \SL \cdot 0$, and a suffix $\SR \cdot \rev{\per}^{k_1}$, so a right hairpin deletion with length $d$ is possible from $v'_2$ for every $d \le k_1 \cdot |\per| +3$.
Note that in order to reach $v_3$ from $v'_2$, a right hairpin deletion with a length strictly less than $k_1 \cdot |\per|$ is required.
Clearly, we have $\cost(p^*) = \cost(p)$ and $\last(p^*) = v'_2 = \per^{k_1} \cdot \SL \cdot q \cdot \SR \cdot \rev{\per}^{k_1} $ as required.
\end{proof}

With that, we are finally ready to prove \cref{lem:FibDistance}.
\FibDistanceLemma*
\begin{proof}
Let $p$ be a shortest path from $v = \per^a\cdot\SL\cdot  q\cdot \SR \cdot \rev{\per}^a$ to $u = q$ in $G_x$.
By \cref{lem:lastkk}, we can assume that $\last(p)= v_1 = \per^k\cdot\SL\cdot  q\cdot \SR \cdot \rev{\per}^k$ for some $k\in[a]$.
Let $p_1$ be the prefix of $p$ from $v$ to $\last(p)$ and let $p_2$ be the suffix of $p$ from $\last(p)$ to $u$.
We analyze the cost of $p_2$ via the following claim.
\begin{claim*}
    $\cost(p_2) = 2 + (1 + \max(\ext-\int - 1,0)) \cdot k$
\end{claim*}  

\begin{claimproof}
We assume that the operation following $v_1$ in 
$p$ is a right hairpin deletion.
The case in which the operation following $v_1$ is a left hairpin deletion is symmetrical.

Let $v_2$ be the vertex that immediately follows $v_1$ in $p$.
From $v_1= \last(p)$ we have that $v_2 = \per^{k} \cdot \SL \cdot q \cdot Y$ for some proper prefix $Y$ of $\SR \cdot \rev{\per}$.
From \cref{clm:shortdeletions,lem:ret0110}, we have that for every $j \in [0.. k]$, a vertex $w_j \in \per^j \cdot \SL \cdot q * $ is visited by $p_2$.
Consider the subpath from $w_{j+1}$ to $w_{j}$ for some $j \in [k]$ (here, we denote as $w_{j}$ and $w_{j+1}$ as the first vertices in $p$ with this property).
Clearly, the length of this subpath is at least $1$.
Recall that $w_{j+1}$ appears after $\last(p)$ in $p$, so the suffix of $w_{j+1}$ following $q$ does not contain a full occurrence of $\rev{\per}$. 
It follows that if $\ext > \int $, the maximal prefix of $\per$ that can be a reverse (of an inverse) of a suffix of $w_{j+1}$ is $010^{\int+1}$.
Therefore the first left hairpin deletion applied after $w_{j}$ in $p$ is of length at most $\int+3$, which results in a string with a prefix $0^{(\ext+1)-(\int+1)} = 0^{\ext-\int}$ if $\ext > \int$.
It follows from \cref{lem:ret0110} that the next $\ext- \int - 1$ operations would be left hairpin deletions, each of a single $0$ character until $w_{j+1}$ is reached.
Finally, at least one more left hairpin deletion operation is required from $w_{0}$ to delete $\SL$. 
It follows from the above that $\cost(p_2) \ge 2 + (1 + \max(\ext-\int - 1,0) )\cdot k$.

Note that a path $p^*_2$ from $v_1$ to $u$ with this cost is obtainable as follows.
First, $p^*_2$ deletes a suffix of length exactly $k \cdot |\per| + |\SR|$ from $v_1$ (note that the suffix of $v_1$ of this length is exactly the reverse (of the inverse) of the prefix of $v_1$ of the same length).
Then, $p^*_2$ deletes every occurrence of $\per$ to the left of $\SL$ separately by deleting $010^{\int+1}$, and then deletes the remaining $\min(\ext-\int - 1,0)$ zeros until the next occurrence of $\per$ one by one.
Finally,  $p^*_2$ deletes $\SR$ in a single operation, reaching $u=q$.
\end{claimproof}

It follows from the claim that $\cost(p_2) = 2+(1+\max(\ext-\int-1),0)\cdot k$ and it follows from \cref{lem:hairpinfib} that $\cost(p_1) = \Fib^{-1}(\frac{a}{k})$.
So the overall cost of $p$ is $\cost(p) = 2 + (1+\max(\ext-\int-1,0))\cdot k +\Fib^{-1}(\frac{a}{k})$.
Consider an alternative path $p^*$ defined as follows.
\begin{enumerate}
    \item \label{pstartstep:1} Step 1: The prefix of $p^*$ is a shortest path from $v$ to $\per \cdot  \SL \cdot q \cdot \SR \cdot \rev{\per}$ in $G_x$.
    
    \item \label{pstartstep:2} Step 2:  $p^*$ applies one right\footnote{There exists another path of the same length starting with left hairpin deletion operation.} hairpin deletion operation of length $|\per| + |\SR|$ to delete $\SR \cdot \rev{\per}$.

    \item \label{pstartstep:3} Step 3: $p^*$ applies a left hairpin deletion to delete $010^{\min(\ext,\int+1)}$. 
    Resulting in the string $0^{\max\{\ext-\int-1,0\}}\cdot\SL\cdot q$.
    
    \item \label{pstartstep:4} Step 4:  $p^*$ applies $\max(\ext-\int-1,0)$ left hairpin deletion operations, each removing a single copy of $0$ from the left, reaching the string $\SL\cdot q$.
    
    \item \label{pstartstep:5}Step 5: $p^*$ applies a length $|\SL|$ length hairpin deletion to delete $\SL$ resulting in $q$. 
\end{enumerate}

According to \cref{lem:fib} the cost of \cref{pstartstep:1} is $\Fib^{-1}(a)$.
The cost of \cref{pstartstep:2,pstartstep:3,pstartstep:4,pstartstep:5} is $\max(\ext-\int-1,0) +3$.

In particular, we have 
\begin{align*}
\cost(p^*) &= 3 + \max(\ext-\int-1,0) + \Fib^{-1}(a) \\&\overset{(*)}{\le} 3 + \max(\ext-\int-1,0) + k - 1 + \Fib^{-1}(\frac{a}{k}) \\&\le 
2 + (\max(\ext-\int-1,0) + 1) \cdot k + \Fib^{-1}(\frac{a}{k}) = \cost(p)
\end{align*}
where the first inequality $(*)$ holds due to $a$ being a Fibonacci-regular number.

In conclusion, we have shown that $p^*$ is an optimal path in $G_x$ from $\per^a \cdot \SL \cdot q \cdot \SR \cdot \rev{\per}^a$ to $q$ with cost $3 + \max(\ext-\int-1,0) + \Fib^{-1}(a)$, as required.
\end{proof}

\end{document}